%% file: main.tex
\pdfoutput=1
\documentclass[conference]{IEEEtran}
\IEEEoverridecommandlockouts
\usepackage[utf8]{inputenc}

\usepackage{cite}
\usepackage{amsmath,amssymb,amsfonts}
\usepackage{algorithmic}
\usepackage{graphicx}
\usepackage{textcomp}
\usepackage{pifont}
\usepackage{amsfonts}
\usepackage[table,xcdraw]{xcolor}
\usepackage{booktabs} 
\usepackage{colortbl}
\usepackage[a4paper, total={184mm,239mm}]{geometry}
\usepackage{capt-of}   
\usepackage{float}     
\usepackage{xcolor}

\usepackage{amsmath}
\usepackage{amsthm}
\usepackage{url} 
\usepackage[switch]{lineno}
\usepackage{multirow} 
\usepackage{makecell}   
\usepackage{enumitem}

\usepackage{placeins}
\usepackage{cuted}
\usepackage{tikz}
\usetikzlibrary{shapes.geometric, arrows.meta, positioning, calc}
\usepackage{microtype}
\newcommand{\change}[1]{\textcolor{black}{#1}}

\newtheorem{definition}{Definition}
\newtheorem{theorem}{Theorem}
\newtheorem{lemma}{Lemma} 

\usepackage{ifthen}
\newif\ifDraftMode \newcommand{\setDraftMode}[1]{\ifthenelse{\equal{#1}{yes}}{\DraftModetrue}{\DraftModefalse}}

\setDraftMode{yes}
\newcommand{\thickhat}[1]{\mathbf{\widehat{\text{$#1$}}}}

\usepackage{lipsum}
\def\BibTeX{{\rm B\kern-.05em{\sc i\kern-.025em b}\kern-.08em
    T\kern-.1667em\lower.7ex\hbox{E}\kern-.125emX}}
\usepackage{algorithm2e}
\begin{document}

\title{Uni-SFU: Algorithm-HW Co-Design for Universal SFUs via Mixed-Degree Piecewise Approximation}

\author{
Miao~Sun,
Yucheng~Huang,
Mingcong~Cao,
Jaehyun~Park,
Partha~Pratim~Pande,
and~Umit~Y.~Ogras
\thanks{Miao Sun and Partha Pratim Pande are with the School of Electrical Engineering and Computer Science, Washington State University, Pullman, WA 99164 USA (e-mail: miao.sun@wsu.edu; pande@wsu.edu).}%
\thanks{Yucheng Huang, Mingcong Cao, and Umit Y. Ogras are with the Department of Electrical and Computer Engineering, University of Wisconsin--Madison, Madison, WI 53706 USA (e-mail: uogras@wisc.edu).}%
\thanks{Jaehyun Park is with the Department of Electrical, Electronic and Computer Engineering, University of Ulsan, Ulsan 44610, Republic of Korea (email:jaehyun@ulsan.ac.kr).}%
\thanks{This work was supported by the US National Science Foundation (NSF) grant CSR-2308530, the Army Research Office under Grant ARO-W911NF-24-1-0240.}%
}

\maketitle

\begin{abstract}
Nonlinear activation functions are essential to modern deep neural networks (DNNs), but their hardware evaluation places significant pressure on the special-function units (SFUs) of GPUs and \change{custom accelerators. Therefore, piecewise polynomial approximations are commonly used within allowed error bounds to improve computational efficiency.} However, existing techniques often approximate each \change{activation} function in isolation using fixed-degree polynomials and uniform segments, leading to hardware redundancy and sub-optimal precision. To address these limitations, we present Uni-SFU, an algorithm-hardware co-design framework that jointly optimizes approximation accuracy and silicon area for a diverse set of activation functions. Uni-SFU leverages a joint search across all target functions to assign mixed-degree polynomials to nonuniform segments, guided by an RTL-derived area cost model. 
\change{This approach identifies a unified hardware configuration to implement the target activation functions under given accuracy constraints.}
Validated across over 700 neural network variants and three Natural Language Processing (NLP) models, Uni-SFU achieves a superior Mean Squared Error (MSE) below $8.22\times10^{-8}$, limiting top-1 accuracy degradation to within 1.02\% compared to floating-point baselines. \change{The proposed design occupies only $6,800$~$\mu\text{m}^2$ in GF~22nm CMOS technology}, achieving a superior trade-off between silicon area and system-level accuracy compared to SOTA counterparts. 
\end{abstract}



\input{text/I-introduction-sfu}
\input{text/II-background_and_related_works-sfu}
\input{text/III-Uni-SFU}
\input{text/IV-Hardware-Implementation}
\input{text/V-Implementation_and_Evaluation-sfu}
\input{text/VI-Conclusion-sfu}

\vspace{6mm}
\noindent \textit{Disclosure}: Dr. Ogras is affiliated with Samsung Austin Research \& Development Center and Advanced Computing Lab (SARC/ACL). This relationship has been approved under applicable outside activities policies.
\clearpage
\bibliographystyle{unsrt}
\bibliography{references/ref}

\end{document}

%% file: text/I-introduction-sfu.tex
\section{Introduction} \label{sec:introduction}

An increasingly diverse set of nonlinear activation functions plays a crucial role in the success of DNNs.
They enable AI/ML models to learn complex decision boundaries and data distributions essential to tasks, such as image classification~\cite{krizhevsky2012imagenet}, natural language understanding~\cite{devlin2019bert}, and generative modeling~\cite{ho2020denoising}.
Early non-linear activation functions like ReLU~\cite{nair2010rectified} apply a simple threshold on input features. 
As the AI/ML models continue to evolve, they employ more complex and multiple variants of non-linear activation functions. 
For example, Transformer-based models~\cite{vaswani2017attention} employ Gaussian Error Linear Unit ($\mathrm{GELU}$)~\cite{hendrycks2016gaussian} and its variants, Mamba~\cite{gu2024mamba} integrates Sigmoid Linear Unit ($\mathrm{SiLU})$ and $\mathrm{Swish}$,
while reinforcement learning (RL) agents~\cite{schulman2017proximal} often employ $\mathrm{Tanh}$ and $\mathrm{Sigmoid}$.
These complex activation functions offer smooth probabilistic characteristics and improved training stability, but require $\mathbf{3\times}$ to $\mathbf{8\times}$ more $\text{FLOPs}$ than $\text{ReLU}$. 

As activation functions become increasingly complex, their efficient evaluation in silicon becomes a critical bottleneck for neural network accelerators~\cite{ozen2022architecting, maity2021chauffeur}. Full-precision floating-point evaluation of transcendental functions such as exponentials and logarithms is prohibitively expensive in GPU or ASIC cores, making SFUs indispensable~\cite{nvdiaSFUExplain, vulchi2025hyppo}. Even in specialized quantization-based architectures like SwiftTron~\cite{marchisio2023swifttron}, activation functions often require dequantization and complex rescaling to restore the necessary precision for non-linear operations such as GELU and Softmax. \change{Indeed, recent trends show that} the growing imbalance between Tensor Core throughput and the multi-function unit (MUFU) capacity of GPUs has made activation evaluation a dominant pipeline bottleneck, as reported in FlashAttention-4~\cite{zadouri2026flashattention}. Hence, SFU designs must co-optimize approximation accuracy, silicon area, performance and power consumption 
simultaneously. Furthermore, modern GPU and ASIC accelerators employ multiple parallel processing threads to sustain high throughput, with each thread independently evaluating activation functions on its assigned data stream~\cite{jouppi2017datacenter, liao2021ascend}. 
However, naively replicating the full coefficient \change{memory for each thread increases the storage requirements proportionally~\cite{chen2016eyeriss, choquette20213}.}


Existing hardware approximation methods for activation functions include iterative approaches \change{(e.g., CORDIC~\cite{volder2015cordic} and Newton--Raphson~\cite{markstein2004software}), global expansion-based methods (e.g., Taylor series~\cite{nilsson2014hardware} and Chebyshev polynomials~\cite{nicolas2024chebyshev}), and piecewise polynomial approximation (PPA)~\cite{dong2020plac}. 
Each approach offers different trade-offs in accuracy, area, and function coverage, as detailed in Section~\ref{sec:II}.} 
Among them, PPA has received the most attention in recent SFU hardware designs due to its compatibility with lookup table (LUT)-based datapaths and its ability to support a wide range of activation functions through offline coefficient preloading~\cite{reggiani2023flex, prasad2025pace, geng2023qpa}.

Current piecewise approximation-based SFU designs suffer from two major inefficiencies. First, they approximate each segment with a uniform degree of polynomial, \change{but regions with low curvature (e.g., the tails of SiLU) can be captured accurately with linear or constant polynomials, while other regions may require higher-degree polynomials~\cite{reggiani2023flex, prasad2025pace, prasad2025lite}.} Therefore, this uniform provisioning wastes both coefficient storage and compute resources on segments that do not require high-order approximation~\cite{dong2020plac, lyu2021ml}. Second, and more critically, existing methods optimize each activation function independently~\cite{geng2023qpa, lu2023efficient, zheng2026ufp}. In practice, a single SFU must serve multiple activation functions across different layers of the same accelerator, since different layers in modern DNN accelerators are configured with different nonlinear activation functions~\cite{vaswani2017attention, devlin2019bert}. Hence, isolated per-function optimization leads to poor area and accuracy trade-offs. For example, a function deemed ``simple'' in isolation may under-utilize expensive compute resources already required by a co-deployed complex function~\cite{reggiani2023flex}. Jointly determining a single hardware configuration that minimizes area while satisfying the accuracy requirements of \textit{all target functions simultaneously} remains an open problem.

To address this problem, we propose Uni-SFU, an \textit{Algorithm-Hardware Co-Design} Framework for designing efficient universal SFUs targeting both server and edge-computing systems. 
\change{The proposed framework jointly explores hardware resources across multiple activation functions and can be adapted to different nonlinear functions through offline coefficient generation.}
\change{Uni-SFU optimizes the \textit{nonuniform breakpoint} locations and assigns variable polynomial degrees per segment using a novel \textit{mixed-degree polynomial} approximation algorithm.}
It jointly \change{co-designs} the approximations of all activation functions and hardware, directly resolving the conflicts among approximation accuracy, memory footprint, and execution latency. 
Furthermore, Uni-SFU directly incorporates a silicon area cost model derived from RTL synthesis rather than optimizing a proxy metric such as the number of segments, in contrast to the current SOTA counterparts.

In current neural network processors, the output from one layer is typically partitioned into data groups. For example, a Streaming Multiprocessor (SM) in a GPU can generate as many outputs per cycle as its CUDA core count. 
To process multiple inputs in parallel efficiently, Uni-SFU adopts a novel multi-lane processing architecture. It replicates the Breakpoint Comparator and Execution Unit for each lane and stores the most recently used coefficients in a dedicated per-lane Local Cache. 
These per-lane hardware units share a Global Multi-Level LUT, which stores all of the segments and polynomial coefficients required by the active activation function. 
Sharing the storage across multiple layers saves area, while the per-lane Local Cache optimizes the end-to-end latency and throughput, maintaining near-peak throughput. 

Uni-SFU is validated across six popular activation functions using over 700 neural network variants from the TIMM library~\cite{rw2019timm}, covering diverse CNN and Transformer architectures on ImageNet-1k~\cite{imagenet15russakovsky}. Experimental results demonstrate that Uni-SFU achieves a maximum MSE of $8.22\times10^{-8}$ across all supported functions. It occupies $6,800$~$\mu\text{m}^2$ in GF 22nm CMOS technology. This compact footprint is achieved by co-optimizing non-uniform breakpoints and mixed-degree polynomials \change{for all target activation functions}, eliminating hardware redundancy without compromising approximation precision.
Uni-SFU employs a fully pipelined architecture, sustaining a peak throughput of 32-bit output per cycle per lane. The execution latency for an individual request varies from 11 to 16 cycles, depending on the polynomial degree and Local Cache hit status.
When integrated with state-of-the-art (SOTA) neural network models, the end-to-end inference accuracy degrades by less than $1.02\%$ compared to floating-point baselines, validating the effectiveness of the proposed co-design approach.

The main contributions of this paper are as follows:
\begin{itemize}[leftmargin=*]

\item \textbf{An Algorithm--Hardware Co-Design Framework}
that simultaneously approximates multiple activation functions using non-uniform breakpoint search and mixed-degree polynomial assignment, guided by an RTL-driven area model.

\item \textbf{A Universal Reconfigurable Hardware Architecture}
that supports six piecewise-approximable activation functions under FP32 precision. Its multi-lane extension introduces a Local Cache that exploits high input locality across lanes to achieve significant area efficiency. 

\item \textbf{Hardware Implementation and End-to-End Validation} 
in GF 22nm technology. Comprehensive end-to-end validations across 700 neural network models and three NLP models show Uni-SFU robustness and high inference fidelity across a diverse range of mainstream DNN workloads.
\end{itemize}

The remainder reviews related work (Section~\ref{sec:II}), the co-design framework (Section~\ref{sec:III}), hardware architecture (Section~\ref{sec:IV}), experimental results (Section~\ref{sec:V}), and concludes in Section~\ref{sec:VI}.



%% file: text/II-background_and_related_works-sfu.tex
\section{Related Work} \label{sec:II}

Hardware approximation of nonlinear activation functions falls into three broad families: iterative methods, expansion-based methods, and piecewise polynomial approximation.

\textbf{Iterative methods} such as CORDIC~\cite{volder2015cordic} and Newton--Raphson iteration~\cite{markstein2004software} compute transcendental functions through successive refinement. Although hardware-friendly, CORDIC suffers from slow convergence and long computation cycles that limit throughput. Similarly, the Newton--Raphson method requires multiple multiplications per iteration, resulting in a long critical path and reducing operating frequency~\cite{hong2024reconfigurable}. However, such iterative methods are not universally applicable; for instance, CORDIC natively supports only a fixed set of trigonometric and hyperbolic functions and fails to directly accommodate modern activations like GELU or SiLU without significant restructuring.

\textbf{Expansion-based methods} fit a global polynomial over the entire input domain. The Taylor series~\cite{nilsson2014hardware} offers a simple coefficient structure but rapidly degrades away from the expansion point and becomes numerically unstable for wide input ranges. The Chebyshev approximation~\cite{nicolas2024chebyshev} achieves a near-minimax error distribution and has been adopted in recent hardware designs. For instance, UFP~\cite{zheng2026ufp} uses a third-degree Chebyshev framework with dynamic-programming segmentation to support a wide range of nonlinear functions under integer arithmetic. Despite their accuracy advantages, global expansion methods must use a sufficiently high polynomial degree to capture the steepest regions of the target function (e.g., the elbow of GELU near zero), imposing the same high multiplier cost uniformly across regions where a constant or linear approximation suffices.

\textbf{Piecewise polynomial approximation} partitions the input domain into segments and fits a low-degree polynomial to each. This allows hardware complexity to be matched to local function behavior: steep regions use higher-degree polynomials, while smooth regions use constant or linear fits, directly reducing multiplier and adder counts. PLAC~\cite{dong2020plac} established the foundation for error-flattened PWL approximation, while ML-PLAC~\cite{lyu2021ml} extended this to multiplier-less hardware via shift-and-add operations. More recently, designs like MARCA~\cite{li2024marca} have employed function-specific nonuniform fitting to improve accuracy for specialized accelerators. However, such approaches often rely on heuristic partitioning for a limited set of functions, lacking a unified optimization framework that systematically balances area-accuracy trade-offs across a diverse library of activations. 
Consequently, a gap remains in achieving a universal and area-efficient SFU \change{design} that is robust across massive-scale neural network workloads. 
PPA maps naturally onto a unified hardware architecture, where a single LUT-based coefficient store and fixed MAC-ADD datapath can support diverse activation functions \change{by reloading parameters at runtime.}
Based on how the approximation is applied to hardware, existing SFU designs can be further divided into \textit{dedicated} and \textit{unified} architectures, as summarized in Table~\ref{tab:sfu_feature_comparison}. Prior work on dedicated SFU architectures has largely focused on individual activation functions. For example, GELU-MSDF~\cite{taghavizade2024gelu} proposes a dedicated GELU accelerator that exploits multi-digit signed representations to reduce gate count, but the design is intrinsically tied to the GELU function and cannot be extended to other activations. Li et al. \cite{li2023high} propose a high-speed SFU for Softmax and GELU that replaces standard arithmetic with efficient shift-and-add operations. Although it delivers 102.4~Gbps at 200~MHz, the hardware is non-reconfigurable, limiting its use to a fixed set of functions. Similarly, PEANO-ViT~\cite{sadeghi2024peano} co-locates dedicated layer normalization, softmax, and GELU blocks on a single ViT accelerator via different approximation methods, \change{but each unit is still independently provisioned,}
precluding any cross-function sharing or reconfiguration. While these approaches achieve high efficiency for their target functions, they ultimately produce distinct physical hardware blocks for each algorithm \textit{rather than a single shared Execution Unit}, and cannot be easily reused across models with varying nonlinearity requirements.

To improve flexibility, subsequent research has explored unified SFU architectures supporting multiple activations simultaneously using PPA. QPA~\cite{geng2023qpa} proposes a quantization-aware PPA methodology that uses the Remez algorithm for minimax fitting and assigns per-multiplier coefficient bit widths to reduce hardware cost. However, each activation function is optimized independently without a unified hardware architecture that jointly supports multiple activation functions on shared resources. Flex-SFU~\cite{reggiani2023flex} supports linear and quadratic degrees with nonuniform breakpoints, while PACE~\cite{prasad2025pace} generalizes this approach to allow any polynomial degree. However, both techniques enforce a single uniform degree across all segments of a given function. 
\change{Similarly, UFP~\cite{zheng2026ufp} adopts a Chebyshev polynomial framework with dynamic-programming segmentation but uses a single degree for all segments.} 
\vspace{-1mm}
\section{Design Challenges and Proposed Solutions}
\label{sec:III}

\change{Existing solutions fail to exploit synergies across diverse activation functions and address the following three design challenges, as summarized in Table~\ref{tab:sfu_feature_comparison}.}

\newcommand{\cmark}{\ding{51}} 
\newcommand{\xmark}{\ding{55}} 

\begin{table}[t]
\centering
\caption{Comparison of SFU Design Features.}
\label{tab:sfu_feature_comparison}
\setlength{\tabcolsep}{3.5pt} 
\renewcommand{\arraystretch}{1.15}
\begin{tabular}{l@{\hspace{6pt}} c c c}
\toprule
\textbf{Feature} & 
\makecell[c]{\textbf{Dedicated}\\\textbf{SFU}\\\scriptsize[34]--[37]} & 
\makecell[c]{\textbf{Unified}\\\textbf{SFU}\\\scriptsize[24]--[27], [29], [30]} & 
\textbf{Ours} \\
\midrule
Nonuniform segments           & \xmark & \cmark & \cmark \\
Extensible to new functions   & \xmark & \cmark & \cmark \\
Optimized mixed degree        & \xmark & \xmark & \cmark \\
Exploit locality              & \xmark & \xmark & \cmark \\
Algorithm--HW co-optimization & \cmark & \xmark & \cmark \\
\bottomrule
\end{tabular}
\end{table}

\setlength{\textfloatsep}{12pt}
\noindent\textbf{Challenge 1: Inefficiency of uniform polynomial degrees.}
While Flex-SFU~\cite{reggiani2023flex} utilizes stochastic gradient descent to enable nonuniform breakpoint locations, it remains restricted to a uniform polynomial degree (linear or quadratic) across all segments. Similarly, PACE~\cite{prasad2025pace} and PACE-Lite~\cite{prasad2025lite} allow degree selection per function but enforce that the chosen degree is uniform across every segment. This creates a resource-accuracy mismatch since a high-degree polynomial is required for high-curvature regions, but results in significant hardware overhead when applied to nearly linear segments. The challenge lies in simultaneously optimizing nonuniform breakpoint locations and per-segment degree to minimize the total area and error.

\noindent\textbf{Challenge 2: Multi-function provisioning without joint optimization.}
Existing frameworks like QPA~\cite{geng2023qpa}, NPLA~\cite{lu2023efficient}, and UFP~\cite{zheng2026ufp} determine coefficient bit-widths, segment counts and polynomial degrees for each target function in isolation. This leads to an architectural optimization gap in unified hardware contexts, where the SFU must be sized for the most numerically demanding function, causing redundant area for simpler ones. The challenge is to perform a joint cross-function hardware search that balances the disparate area-to-accuracy sensitivities of multiple nonlinear functions within a single, unified resource budget.

\noindent\textbf{Challenge 3: Area scaling in parallel multi-lane coefficient storage.}
As Transformer models scale, the area footprint of SFUs becomes the primary bottleneck due to linear coefficient replication. Current designs such as PACE-Lite~\cite{prasad2025lite} and Hong et al.~\cite{hong2024reconfigurable} primarily focus on single-lane or isolated-lane optimizations and do not account for resource reuse or data locality across multiple lanes, leading to the total area scaling linearly with lane count. While Flex-SFU~\cite{reggiani2023flex} attempts to share coefficient storage through multiple access ports, the resulting routing congestion and port contention scale poorly with massive parallelism. The critical architectural challenge is addressing this linear area growth by exploiting input locality across lanes to enable high-efficiency coefficient sharing.

We address the above-mentioned challenges using an algorithm-hardware co-design framework that jointly resolves all three limitations within a unified SFU design.  
On the algorithmic side, Uni-SFU employs a mixed-degree piecewise polynomial approximation with nonuniform breakpoints (Section~\ref{subsec:IV-B},~\ref{subsec:IV-C}). 
Furthermore, it performs a joint cross-function search to identify the hardware configuration with minimal area that meets the accuracy requirements for all target functions (Section~\ref{subsec:IV-D} and ~\ref{subsec:IV-E}). On the hardware side, a multi-lane architecture with a Local Cache exploits the high input locality of activations to eliminate redundant coefficient replication across lanes, breaking the linear area scaling bottleneck (Section~\ref{sec:V}).


%% file: text/III-Uni-SFU.tex
\section{Uni-SFU Algorithm-HW Codesign Framework}\label{sec:IV}
\subsection{Overview}\label{subsec:IV-A}
The inputs to the Uni-SFU framework are a list of target nonlinear activations to be supported by the optimized hardware and the allowed global maximum error, as shown in Figure~\ref{fig:overview}. 
The novel aspects of this step are 1) identifying the optimum set of nonuniform segments and 2) approximation of each segment using mixed-degree polynomials. 
For instance, the illustrative example in Figure~\ref{fig:overview} divides the SiLU into 11 nonuniform segments and approximates 2 of them with constant (Z), 4 with linear (L), 2 with quadratic (Q), and 3 with cubic (C) polynomials.

\change{The first step of the proposed framework finds \textit{all Pareto-optimal solutions} that meet the global maximum error constraint using a Dynamic Programming (DP) algorithm.
These per-function Pareto sets form the interface between this stage and the subsequent cross-function best-first search. 
The second step uses \change{per-function Pareto sets} to identify the hardware configuration that meets the \textit{maximum error constraint across all functions}. 
In summary, the DP algorithm in the first stage and the Best-first Dijkstra search in the second stage operate sequentially to solve different aspects of the overall optimization problem. The DP stage is applied to each activation function independently to filter the design space and generate a Pareto-optimal candidate set. 
The Best-first Dijkstra search then traverses the graph formed by discrete Pareto sets to search for a joint hardware implementation that satisfies all target functions under user-specified error constraints with the minimum total area.}

\change{A key novelty of the second} step is co-optimizing the error of each target function simultaneously using hardware area estimates obtained by our RTL implementation and synthesis with GF 22 nm technology. 
In contrast, existing approaches use the number of segments as a proxy for hardware area, which is an invalid assumption~\cite{prasad2025pace, prasad2025lite}. 

\begin{figure}[t]
\centering
\includegraphics[width=0.9\linewidth]{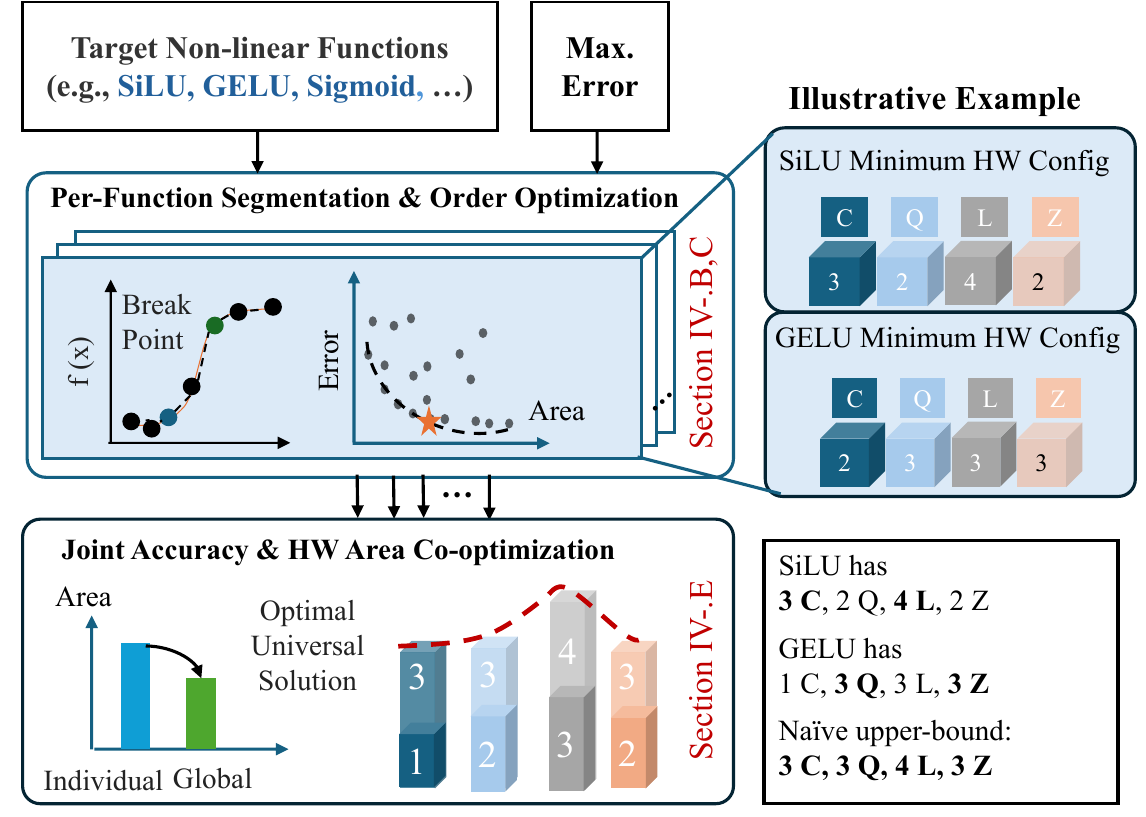}
\vspace{-2mm}
\caption{The proposed framework takes target nonlinear activation functions as input, derives per-function minimum hardware configurations Section~\ref{subsec:IV-C}, and merges them into a single universal configuration in Section~\ref{subsec:IV-E} that is smaller than the na\"{i}ve element-wise maximum.}
\label{fig:overview}
\end{figure}


\begin{figure*}[t]
\centering
\includegraphics[width=0.90\linewidth]{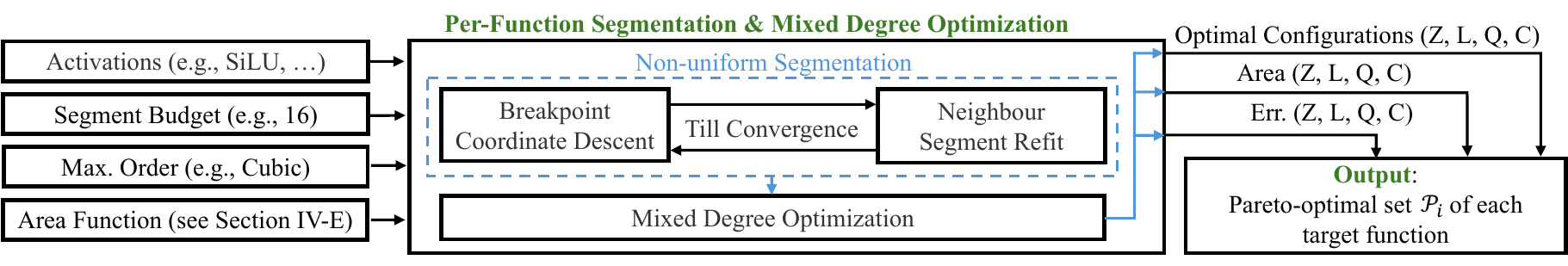}
\vspace{-3mm}
\caption{Description of the Per-Function Segmentation \& Mixed Degree Optimization within the Uni-SFU HW/SW co-optimization framework (Section~\ref{subsec:IV-C}).}
\label{fig:sw-stage1}
\vspace{-4mm}
\end{figure*}

\vspace{-2mm}
\subsection{Preliminaries} \label{subsec:prelim}
\label{subsec:IV-B}

Let $\mathcal{S}_A$ be the set of target activation functions and
$[x_{\min},x_{\max}]$ be their common domain.
We denote the $i^\mathrm{th}$ activation function by
$y(i,x)\in\mathcal{S}_A$,
where $i\in\mathbb{Z},\ 1\le i\le |\mathcal{S}_A|$,
and $x\in[x_{\min},x_{\max}]$.

Suppose a function is divided into $K$ nonuniform segments
denoted by $y_j(i,x)$, $i, j \in \mathbb{Z} , 1 \leq j \leq K$.
The $x-$coordinates of the breakpoints are given by 
$\{x_{min}, x_1, \dots, x_{K-1}, x_{max}\}$, 
such that the first segment starts at $x_{\min}$ and the last segment ends at
$x_{\max}$. 
The proposed Uni-SFU framework finds the intermediate points $\{x_1, \ldots, x_{K-1}\}$. 

A given segment could be approximated with a cubic ($3^\mathrm{rd}$-degree),
quadratic ($2^\mathrm{nd}$-degree), linear ($1^\mathrm{st}$-degree), or 
constant polynomial. 
$\thickhat{y}_{j,d}(i,x)$ denotes the approximation for segment $j$ of the $i^\mathrm{th}$ function with degree $d\in\{0,1,2,3\}$.
The corresponding MSE between the original segment and its approximation is:
\begin{equation} \label{eqn:err}
    e_{j,d}(i) = \frac{1}{x_{j} - x_{j-1}} \int_{x_{j-1}}^{x_{j}} \big(y_j(i,x) - \thickhat{y}_{j,d}(i,x)\big)^2 \,dx
\end{equation}
for $1 \leq j \leq K$, where $x_0 = x_{min}$ and $x_K = x_{max}$.

Since each segment can be independently approximated with a 0 to $3^\mathrm{rd}$ degree polynomial, 
there are $4^K$ possible combinations. 
A given approximation can be uniquely specified by the \textbf{degree tuple} defined as follows.

\begin{definition} 
\textbf{Degree Tuple} is defined as the tuple
$$\mathbf{d} = (d_1, d_2, \ldots, d_K), \quad d_j \in \{0, 1, 2, 3\} $$
where $d_j$ is the degree of the polynomial that approximates segment $j \in \{1,2,\ldots,K\}$.
\end{definition}

Different degree tuples (i.e., implementations) may lead to the same hardware cost. 
For example, $\mathbf{d}=(1,2,3)$, $\mathbf{d}=(2,1,3)$, and $\mathbf{d}=(3,2,1)$ all imply one linear, one quadratic, and one cubic segment. 
\change{They have the same area cost} because the hardware depends only on how many segments use each polynomial degree, regardless of the specific segment locations.  
To exploit this symmetry and have a compact representation that can be used in a search tree, 
we define the marginal and cumulative degree distributions as follows.

\begin{definition} \label{def:degree_dist} 
The \textbf{Marginal Degree Distribution} is defined as the tuple
\[
\mathbf{m_d} = (m_Z, m_L, m_Q, m_C),
\]
where $m_Z$, $m_L$, $m_Q$, and $m_C$ denote the number of segments fitted using constant (zero-degree), linear, quadratic, and cubic polynomials, respectively.
\end{definition}

\begin{definition} \label{def:Cumu_degree_count} 
The \textbf{Cumulative Degree Distribution} is defined as the tuple $\mathbf{c_d} = (c_{d1}, c_{d2}, c_{d3})$, where $c_{d3}$ is the number of segments fitted using degree 3, 
$c_{d2}$ is the number of segments fitted using degree 2 or 3,
$c_{d1}$ is the number of segments fitted using degree 1 or higher.
$c_0=K$ is implicit since all segments are of degree 0 or higher, and all have a constant term.
\end{definition}

\subsection{Per-Function Segmentation \& Mixed Degree Optimization}
\label{subsec:IV-C}
\noindent \textbf{Nonuniform segmentation:} 
The first step of Uni-SFU is to find the optimum nonlinear segmentation for each nonlinear activation function, 
i.e., the locations of the breakpoints $\{x_1, \ldots, x_{K-1}\}$ (recall that $x_{min}$ and $x_{max}$ are fixed). 
To this end, it first initializes the breakpoints uniformly and assumes that the highest allowed degree is used across all segments\footnote{It is cubic in this work since using higher degree polynomials do not provide significant benefits. The proposed algorithm can also be used with higher degrees. Mixed degrees are assigned in the second step.}. 
This starting point is a commonly used segmentation approach, but it is not necessarily optimum. 
\change{Therefore, Uni-SFU applies coordinate descent optimization over breakpoint locations to iteratively shift each breakpoint $x_j$ within the range $[x_{j}, x_{j+1}]$, as outlined in Figure~\ref{fig:sw-stage1}.}
The first iteration starts by finding the $x_1$ that minimizes the fitting error across neighboring segments by moving it on a discretized grid, assuming all other breakpoints remain at their starting positions. 
Then, Uni-SFU moves the second breakpoint $x_2$ to minimize the MSE while freezing the rest of the breakpoints. This local optimization process is applied to each intermediate breakpoint. Then, the entire iteration is repeated until either none of the breakpoints move (i.e., they have converged) or the maximum allowed number of iterations is reached. 
In our evaluations, all functions converged within 20 iterations over the full set of breakpoints.

\noindent \textbf{Mixed Degree Optimization:} 
\change{The mixed-degree optimization step seeks degree assignments that balance approximation error and hardware area. Since the design space is large ($4^K$ for up to cubic degree)}, prior approaches resort to a uniform degree distribution (e.g., all segments are linear/quadratic)~\cite{reggiani2023flex, prasad2025pace, prasad2025lite, lu2023efficient, dong2020plac, lyu2021ml}. 
In contrast, we propose an efficient DP algorithm with the help of the \textit{cumulative degree distribution} (Definition~\ref{def:Cumu_degree_count}).
We use the cumulative degree distribution $\mathbf{c_d}$ as the DP state since each new segment assignment updates it additively and monotonically: a degree-$d$ segment contributes one count to every cumulative level up to $d$. For example, a 2nd degree segment fit would increment $c_{d1}, c_{d2}$ by one but does not increase $c_{d3}$. Hence, $\mathbf{c_d}$ provides a compact and natural state representation for the DP recursion.

For a given segment, using higher-degree polynomials yields a lower approximation error.
Therefore, the segment-wise fitting error of segment $j$ satisfies the following relation:
\[e_{j,3}\le e_{j,2}\le e_{j,1}\le e_{j,0}.\]
Thus, assigning cubic polynomials to all segments yields the smallest approximation error, but also incurs the largest hardware cost due to the increased number of coefficients and computations.
In contrast, lower-degree assignments reduce hardware cost but generally increase approximation error.
Consequently, different degree assignments produce different area--error trade-offs. 
The objective of the mixed-degree optimization step is to identify the Pareto-optimal assignments.

\vspace{1mm}
\begin{definition}\label{def:partial_degree_assignment} 
A \textbf{partial degree assignment} after processing the first $j$ segments is denoted as the tuple:
\[
\pi^{(j)}=(d_1,d_2,\ldots,d_j),
\qquad d_\ell\in\{0,1,2,3\}\ \text{for } \ell=1,\ldots,j.
\]
\end{definition}
Its cumulative approximation error is
\[
E(\pi^{(j)})=\sum_{\ell=1}^{j} w_\ell\,e_{\ell,d_\ell},
\qquad
\mathrm{where~} w_\ell=\frac{x_\ell-x_{\ell-1}}{x_{\max}-x_{\min}} 
\]
and $e_{\ell,d_\ell}$ defined in Equation~\ref{eqn:err}. The cumulative degree-count reached by $\pi^{(j)}$ is denoted by $\mathbf{c}(\pi^{(j)})$, as summarized in Table~\ref{tab:combined_notation}.

\vspace{1mm}
\begin{lemma}\label{def:state_dominance} 
Let $\pi_a^{(j)}$ and $\pi_b^{(j)}$ be two partial degree assignments in the same DP layer $j$.
Partial assignment $\pi_a^{(j)}$ \textbf{dominates} $\pi_b^{(j)}$ ($\pi_a^{(j)} \succ_{\mathrm{st}} \pi_b^{(j)}$)
if 
\[
\mathbf{c}(\pi_a^{(j)})=\mathbf{c}(\pi_b^{(j)})
\ \text{and}\
E(\pi_a^{(j)}) < E(\pi_b^{(j)}).
\]
\end{lemma}
\begin{proof}
If $\mathbf{c}(\pi_a^{(j)})=\mathbf{c}(\pi_b^{(j)})$, these partial assignments have equal area cost.
If $E(\pi_a^{(j)}) < E(\pi_b^{(j)})$, assignment $a$ yields lower error at equal area cost and dominates assignment $b$.
\end{proof}

\begin{table}[!t]
\centering
\renewcommand{\arraystretch}{1.2}
\setlength{\tabcolsep}{2pt}
\caption{Notation for Uni-SFU framework.}
\label{tab:combined_notation}
\vspace{-2mm}
\small
\begin{tabular}{p{0.16\columnwidth}<{\centering} p{0.32\columnwidth} | p{0.11\columnwidth}<{\centering} p{0.37\columnwidth}}
\hline
\textbf{Sym.} & \textbf{Description} & \textbf{Sym.} & \textbf{Description} \\
\hline
$\mathcal{S}_A$ & Target function set & $\mathbf{d}$ & Degree tuple \\
{\scriptsize $[x_{\min},x_{\max}]$} & Input domain & $\mathbf{m_d}$ & Marginal degree dist. \\
$i$ & Function index & $\mathbf{c_d}$ & Cumulative degree dist. \\
$y(i,x)$ & $i^{\text{th}}$ target function & $\pi^{(j)}$ & Partial degree assign. \\
$K$ & Number of segments & $E(\pi^{(j)})$ & Cumulative error \\
$j$ & Segment index & $\mathbf{c}(\pi^{(j)})$ & DP state \\
$x_j$ & $j^{\text{th}}$ breakpoint & $\mathbf{v}(d)$ & State increment vector \\
$y_j(i,x)$ & Function on segment $j$ & $V_j(\mathbf{c})$ & DP value function \\
$d$ & Polynomial degree & $\mathcal{P}_i$ & Full Pareto frontier \\
$\thickhat{y}_{j,d}(i,x)$ & Degree-$d$ approx. & $\mathcal{Q}_i$ & Error-feasible subset \\
$e_{j,d}(i)$ & Approx. error & $\mathbf{u}$ & Hardware bound \\
\hline
\end{tabular}
\end{table}

\noindent\textbf{Dynamic Programming Procedure:}
The proposed mixed-degree optimization is formulated as a $K$-horizon sequential decision process. Each step $j \in \{1, \dots, K\}$ corresponds to assigning a polynomial degree to the $j^\mathrm{th}$ segment.
The cumulative degree count $\mathbf{c}(\pi^{(j)})$ is chosen as the state at step $j$, as this avoids an exhaustive $4^K$ search and yields the optimal solution, as proven below.

We first define the states for the DP problem. At each decision step for segment $j$, a polynomial degree $d_j \in \{0, 1, 2, 3\}$ is assigned to the $j^\mathrm{th}$ segment. The state at step $j$ is strictly defined by the cumulative degree distribution reached by the partial assignment $\pi^{(j)}$, denoted as $\mathbf{c}(\pi^{(j)}) = (c_{d1}, c_{d2}, c_{d3})$. When a degree $d_j$ is selected, the state transitions additively: $\mathbf{c}(\pi^{(j)}) = \mathbf{c}(\pi^{(j-1)}) + \mathbf{v}(d_j)$, where $\mathbf{v}(d_j)$ is the increment vector for the cumulative counts (e.g., $\mathbf{v}(2) = (1, 1, 0)$ and $\mathbf{v}(0) = (0, 0, 0)$).

Let $V_j(\mathbf{c})$ denote the minimum cumulative approximation error among all partial degree assignments that reach cumulative state $\mathbf{c}$ after processing $j$ segments, i.e.,
\[
V_j(\mathbf{c})
=
\min_{\pi^{(j)}:\,\mathbf{c}(\pi^{(j)})=\mathbf{c}} E(\pi^{(j)}).
\]
The base case is initialized as $V_0((0,0,0))=0$ and $V_0(\mathbf{c})=\infty$ for all other states. For $1 \le j \le K$, the Bellman recursion is
\[
V_j(\mathbf{c})
=
\min_{d \in \{0,1,2,3\}}
\Big[
V_{j-1}\big(\mathbf{c}-\mathbf{v}(d)\big)+w_j\,e_{j,d}
\Big].
\]
where $w_j=\frac{x_j-x_{j-1}}{x_{\max}-x_{\min}}$ denotes the normalized length of segment $j$.
Any transition resulting in an invalid predecessor state is assigned infinite cost. Starting from $V_0((0,0,0))$, the DP records each reachable state at every layer. For each unique state, the algorithm retains only the partial assignment with the \textit{minimum cumulative approximation error}, storing it alongside a backpointer that records the predecessor state and the selected degree. Upon reaching the final horizon $K$, the algorithm yields a set of valid terminal states, which represent the minimum possible error for each unique hardware configuration. By tracing the backpointers from these terminal states, the algorithm reconstructs the optimal segment-level degree assignments $\mathbf{d}^*$. 

After the terminal DP states are collected, we perform another round of breakpoint movement for all of the segment fits. This process aims to reduce the segment bias introduced during initialization, where non-uniform breakpoints were determined under the assumption that every segment would use the maximum allowable degree. These representative assignments are then evaluated to extract the final Pareto front, providing tradeoffs between hardware area and approximation accuracy.

\begin{theorem}[\textit{Optimal Substructure}]
\label{thm:optimal_substructure}
Let $\pi_j^*$ be an optimal partial degree assignment for the first $j$ segments that reaches a cumulative degree state $\mathbf{c}$ with minimal error $V_j(\mathbf{c})$. If $\pi_j^*$ assigns degree $d_j^*$ to the $j^\mathrm{th}$ segment, then its predecessor $\pi_{j-1}^*$ must be an optimal partial assignment for reaching the predecessor state $\mathbf{c}' = \mathbf{c} - \mathbf{v}(d_j^*)$ at step $j-1$.
\end{theorem}

\begin{proof}
Assume the optimal assignment $\pi_j^*$ reaching state $\mathbf{c}$ has a minimal error $V_j(\mathbf{c}) = E(\pi_{j-1}^*) + w_j\,e_{j, d_j^*}$. If the partial assignment $\pi_{j-1}^*$ were not optimal for its terminal state $\mathbf{c}' = \mathbf{c} - \mathbf{v}(d_j^*)$, there would exist a valid assignment $\hat{\pi}_{j-1}$ reaching $\mathbf{c}'$ with strictly lower error $E(\hat{\pi}_{j-1}) < E(\pi_{j-1}^*)$. Substituting $\hat{\pi}_{j-1}$ into the full assignment yields $\hat{\pi}_j$, with a total error $E(\hat{\pi}_{j-1}) + w_j\,e_{j, d_j^*} < V_j(\mathbf{c})$. By Lemma~\ref{def:state_dominance}, $\hat{\pi}_{j}$ dominates $\pi_j^*$, which contradicts the optimality of $\pi_j^*$. Thus, the Bellman equation optimally constructs the solution for every state $\mathbf{c}$.
\end{proof}

Since the problem exhibits this optimal substructure, finding the minimum error for each cumulative degree state at every step $j$ guarantees that the global minimum error is preserved for every possible hardware profile at the final step $K$.
\begin{figure*}
\centering 
\includegraphics[width=0.95\linewidth]{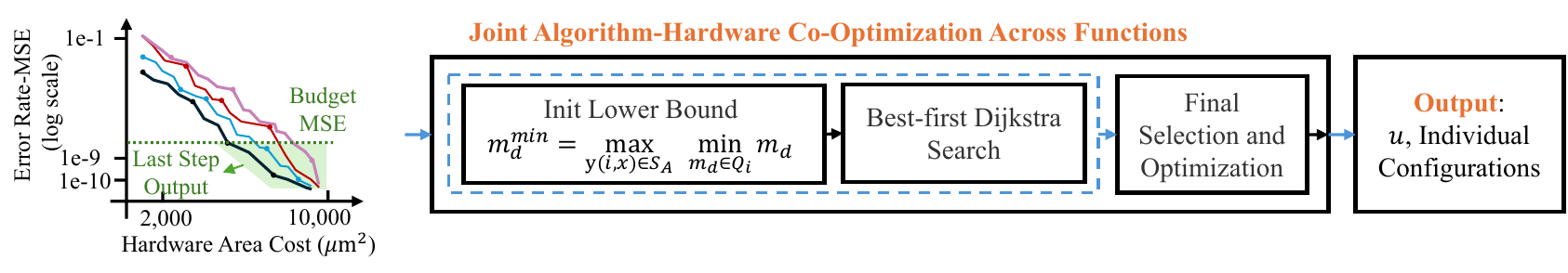} 
\vspace{-5mm} 
\caption{Uni-SFU joint algorithm-hardware co-optimization workflow (Sec.~\ref{ssec:alg-hw-coop}).} 
\label{fig:sw-stage2} 
\vspace{-2mm} 
\end{figure*}

\subsection{Hardware Area Cost Model}\label{subsec:IV-D}
Prior work models hardware cost solely by segment count, which primarily 
governs the coefficient storage area~\cite{reggiani2023flex,gonzalez2021hardware}. 
\change{This abstraction breaks when mixed degree polynomials are used and it fails to account for the compute unit overhead}.

\change{We propose an accurate and broadly applicable area model composed of two groups: (1) the compute unit for polynomial evaluation, and (2) the storage for polynomial coefficients. 
The proposed area model focuses strictly on intrinsic architecture complexity while omitting implementation-specific optimizations such as cache hits (discussed in Section~\ref{sec:V}). Hence, it is  generic and naturally extends to alternative function approximation methods (e.g., Fourier expansion or Newton-Raphson iterations) that can be fundamentally characterized by storage and computing elements.}

\noindent \textbf{Compute Unit Area Model:}
The polynomial evaluation unit consists of serially connected 32-bit MAC and ADD units. To support mixed-order evaluation while maintaining one output per cycle, the number of MAC--ADD pairs is set equal to the maximum polynomial degree present across all segments: one pair for linear, two for quadratic, three for cubic, and a single ADD unit for constant segments. Each configuration is synthesized, yielding $A_{\text{compute}}(\max(\mathbf{d}))$ as a discrete lookup over area cost for different degrees.

\noindent \textbf{Storage Area Model:}
\change{The storage unit uses a multi-level LUT structure (detailed in Section~\ref{subsec:V-B}), where coefficients of different polynomial degrees are stored in separate LUTs to minimize area. Since all segments have a constant coefficient, its LUT has $K$ entries. The LUTs for linear, quadratic, and cubic coefficients have progressively smaller LUTs since fewer segments will need them.}
We sweep the number of segments from 8 to 16 with a step size of 2, including 8, 10, 12, 14, and 16 segments, and enumerate all hardware combinations spanning from entirely constant degrees to entirely cubic degrees for each segment. Out of 2555 unique hardware configurations in this design space, 800 samples are uniformly subsampled and subsequently implemented in RTL, including 50 additional random configurations. All samples are synthesized in the GF 22nm technology, and polynomial regression is applied to yield the fitted base area model:
\vspace{-1mm}
\begin{align}
A_{\mathrm{base}} = \alpha(c_0+c_1) + \beta\,c_2 + \gamma\,c_3 
                  - \delta(c_0 c_1 c_2 c_3) + \epsilon,
\end{align}
where $\alpha, \beta, \gamma, \delta, \epsilon$ are coefficients fitted from 
RTL synthesis results. \change{Consequently, the overall area is modeled as:}
\begin{equation}
A = A_{\mathrm{base}}(c_0,c_1,c_2,c_3) + A_{\text{compute}}(\max(\mathbf{d}))
\end{equation}

\noindent \textbf{Model Validation:}
To verify the fitted model, 30 held-out configurations not seen during fitting are synthesized in GF~22nm and compared against model predictions. The resulting Mean Absolute Percentage Error (MAPE) is 0.83\%, corresponding to an absolute area deviation of $20.19~\mu\text{m}^2$. This error is negligible relative to the tens-of-thousands-of-$\mu\text{m}^2$ scale of the full SFU module, confirming that the model provides sufficient fidelity for the hardware-aware co-design search.

\subsection{Joint Algorithm-Hardware Co-Optimization Across Functions} \label{ssec:alg-hw-coop}
The second step of Uni-SFU operates only on the subset of Pareto-optimal solutions
$\mathcal{P}_i$ that satisfy the maximum error requirement, denoted as
$\mathcal{Q}_i$ (produced by the per-function segmentation \& mixed-degree
optimization step in Section~\ref{subsec:IV-C}). Its goal is to find a
unified hardware bound $\mathbf{u}$ that meets all target functions with
minimum area, as illustrated in Figure~\ref{fig:sw-stage2}.

\noindent\textbf{Lower Bound Universal SFU Hardware:}\label{subsec:IV-E}
As an illustrative example, let us assume the Pareto-optimal set for the function $y(i,x) \in \mathcal{S}_A$ has two possible approximations:
\[
\mathcal{Q}_i = \{(3, 7, 4, 1), (4, 5, 3, 2)\},
\]
where each tuple is ordered as $(Z,L,Q,C)$. That is, the first configuration has $3$ constant, $7$ linear, $4$ quadratic, and $1$ cubic segments, while the second has $4$ constant, $5$ linear, $3$ quadratic, and $2$ cubic segments. A lower-bound marginal degree distribution for this function can be found as
\[
\mathbf{m_d^{min}} = \big(\min(3,4), \min(7,5), \min(4,3), \min(1,2)\big),
\]
\change{since none of the approximations in the Pareto set can fit in a hardware configuration with fewer segments than this bound.}

\change{In general, we must consider a set of activation functions $\mathcal{S}_A$ that need to be implemented by the SFU universal hardware. 
The universal hardware must be large enough to implement any $y(i,x) \in \mathcal{S}_A$. 
That is, the lower bound across all functions is:}
\begin{equation}
\mathbf{m_d^{min}} =\max_{y(i,x) \in \mathcal{S}_A}\ \min_{\mathbf{m_d}\in\mathcal{Q}_i} \mathbf{m_d},
\end{equation}
where each $\mathbf{m_d}\in\mathcal{Q}_i$ is ordered as $(Z,L,Q,C)$, i.e., constant, linear, quadratic, and cubic segment counts.

The lower bound excludes clearly impossible hardware configurations, since any smaller configuration would fail to support at least one target function. 
Therefore, Uni-SFU performs an efficient search in the bounded degree-distribution space starting from the lower bound, as described next.

\noindent\textbf{Algorithm for finding the minimum Universal SFU:}
\change{Each search node is a candidate unified hardware bound $\mathbf{u} = (u_Z, u_L, u_Q, u_C)$, as illustrated in Figure~\ref{fig:s2_flowchart}. The estimated hardware area of each node is used as its priority (the lower the area, the higher the priority).}
The algorithm first inserts $\mathbf{u} = m_d^{min}$ (the current smallest hardware bound) into a priority queue keyed by area.
Then, it repeatedly removes the node with the smallest area and checks whether it can support all target functions, i.e., whether every function has at least one Pareto-optimal candidate that fits within that bound.
If so, the search stops and returns that node as the optimal hardware configuration.
Otherwise, the node is extended by generating up to four neighbors, obtained by increasing one of the entries $u_Z$, $u_L$, $u_Q$, or $u_C$ by one, provided that the new node remains within the upper bound. Then $\mathbf{u}_{\min}$ is set to the smallest area node in the priority queue. These neighbors are then inserted into the priority queue, and the process repeats.

\begin{figure}[t]
\centering
\includegraphics[width=\linewidth]{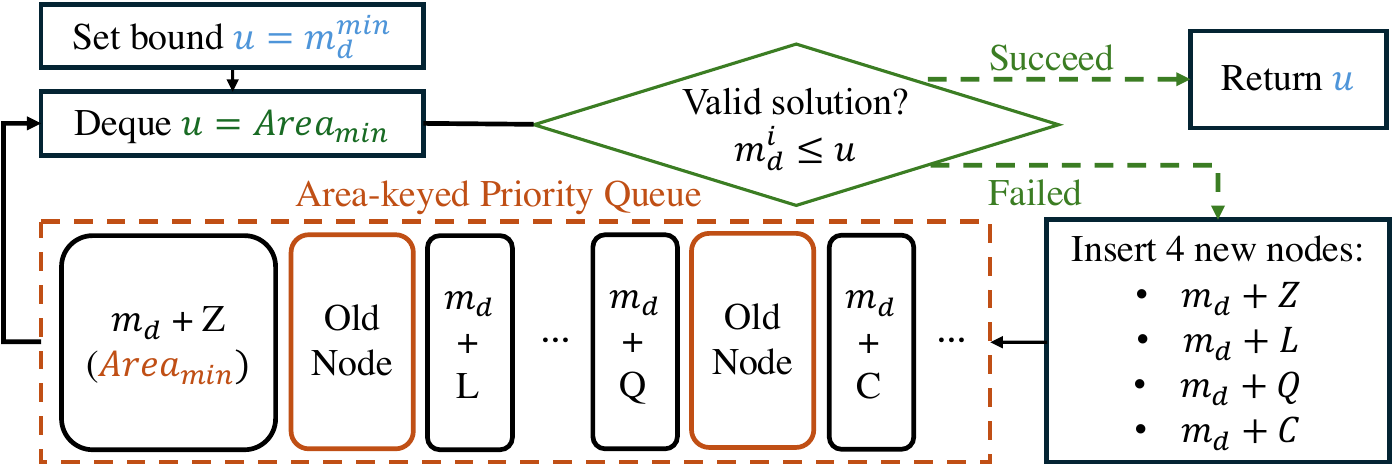}
\vspace{-6mm}
\caption{Flow chart of Best-first Dijkstra Search in the joint algorithm-hardware co-optimization workflow.}
\label{fig:s2_flowchart}
\vspace{-2mm}
\end{figure}

\noindent \textbf{Optimality Analysis:}
This best-first strategy is optimal \change{(within step 2)} and efficient because nodes are explored in increasing hardware area.
Once the first feasible node is found, no unexplored node with smaller area can still satisfy all target functions.
Thus, the first feasible node returned by the search is the minimum-area unified hardware configuration.
After the hardware configuration is determined, Uni-SFU reselects the candidate degree tuple for each target function, which is feasible under the common hardware bound and provides the lowest approximation MSE among all feasible candidates in its Pareto-optimal set.
The selected solution of each function is then further refined by re-optimizing the breakpoint locations while keeping the mixed-degree assignment fixed.
This final refinement improves the approximation quality without changing the hardware configuration obtained by the joint search.

%% file: text/IV-Hardware-Implementation.tex
\section{Proposed Universal SFU (Uni-SFU) Hardware}\label{sec:V}
\begin{figure*}[t]
\centering
\includegraphics[width=0.9\linewidth]{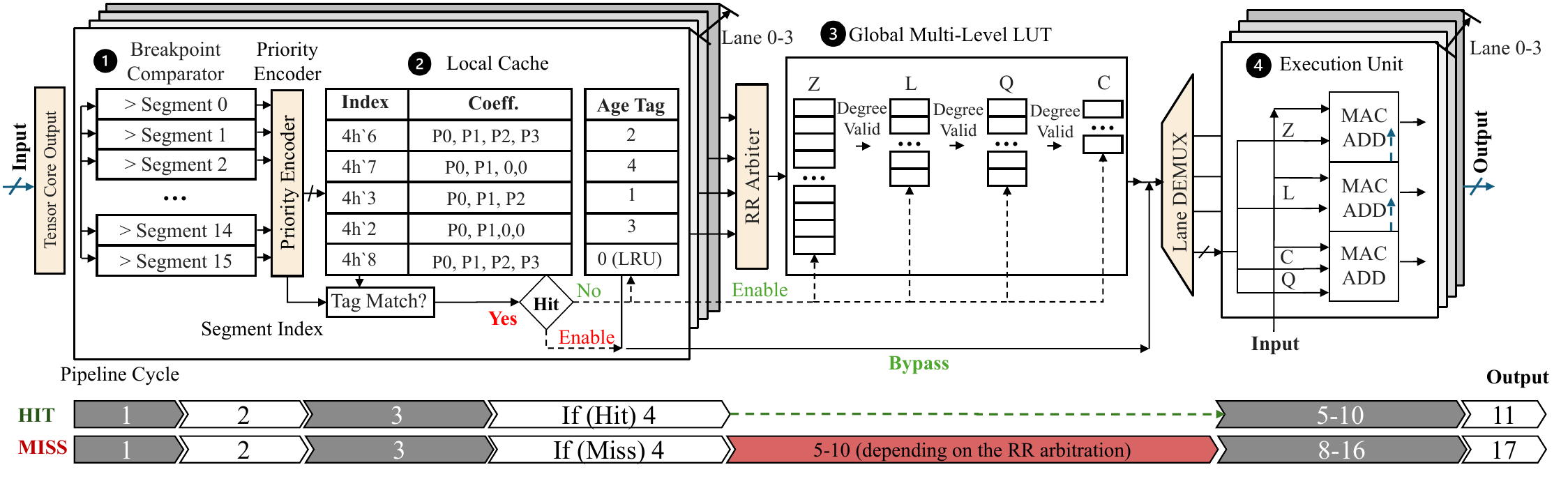}
\vspace{-3mm}
\caption{A four-lane (Lane 0--3) SFU hardware is shown as an example. The details of only Lane 0 are shown for clarity since all lanes have an identical structure.} \label{fig:hw_arch}
\vspace{-3mm}
\end{figure*}

\subsection{Architecture Description and Runtime Operation}\label{subsec:V-A}
The Uni-SFU hardware architecture comprises four units: 
\begin{enumerate}[leftmargin=*]
    \item \textit{Breakpoint Comparator} for active segment identification,
    \item \textit{Local Cache} for storing recently accessed coefficients,
    \item \textit{Global Multi-Level LUTs} for storing all polynomial coefficients for the actively supported nonlinear function, 
    \item \textit{Execution Unit} for 32-bit multiple accumulate operations.
\end{enumerate}

\noindent\textbf{Multi-Lane Operation:} 
Uni-SFU supports a multi-lane operation, as illustrated with \textit{four lanes} in Figure~\ref{fig:hw_arch}.
\change{In this four-lane example, four copies of Breakpoint Comparator, Local Cache, and Execution Unit can process four activation inputs in parallel.} 
Only the Global Multi-Level LUT has a single copy since it has a large footprint (45\% of the total area). 
In this way, Uni-SFU can receive and process as many inputs as the number of lanes in parallel.
The Breakpoint Comparator determines the segments they fall into in the first two cycles, 
while the tag filter checks whether these segments are in the Local Cache in the third cycle.
The segments whose coefficients are in the Local Cache are retrieved from the Local Cache and forwarded to the Execution Unit in the third and fourth cycles (the HIT scenario in Figure~\ref{fig:hw_arch}).
All inputs are processed \textit{in parallel without any stalls} if there is a hit in the Local Cache. 

If any input falls in a segment that is not in the Local Cache (a MISS scenario), the corresponding coefficients must be loaded from the Global Multi-Level LUT. 
These inputs experience a longer latency due to two factors. 
First, access from the Local Cache to the Global Multi-Level LUT requires an arbitration since multiple lanes may experience a miss simultaneously. 
A round-robin arbiter gives each Local Cache access to the Global Multi-Level LUT with a \textit{worst-case contention latency} of three cycles (number of lanes minus one). 
Second, retrieving the coefficients from the Global Multi-Level LUT and forwarding them to the corresponding Execution Unit takes two cycles longer than reading them from the Local Cache. 

Since there are as many Execution Units as the number of lanes, Uni-SFU can output multiple activations (four in our example) in parallel. In summary, one can view the global multi-level LUT as a physical channel shared by multiple Local Caches and Execution Units (like virtual channels). 
Since the global LUT is used only in case of a miss in the Local Cache (not a frequent event), Uni-SFU's performance approaches to the best-case latency and throughput (one output/lane/cycle), as we describe in Section~\ref{subsec:latency}.

\noindent\textbf{Runtime Loading of the Coefficients:} 
Uni-SFU can support an arbitrary number of activation functions, whose segments and coefficients are stored in memory. 
Since only one of them can be active at a time, its coefficients must be loaded to the Global Multi-Level LUT before it is used. 


\subsection{Descriptions of the Hardware Units}\label{subsec:V-B}

\noindent\textbf{Breakpoint Comparator:}
Placed at the input of the Uni-SFU, the Breakpoint Comparator receives the activation generated by the processing elements (e.g., TPU, NPU, or Tensor Core).
For a given nonlinear activation, it preloads the breakpoint locations $x_j$ in registers and keeps them unchanged before another function is enabled. 
It determines the segment index of each input, which also serves as the prefix key for the Local Cache lookup. 
To avoid a long critical path due to cascaded comparison, we evaluate the input against all breakpoints $x_j$ in parallel. Each comparator asserts logic high if the input exceeds $x_j$. In the subsequent cycle, a standard Priority Encoder directly identifies the highest asserted bit to determine the binary segment index. 
Breakpoint comparison and priority encoding are fully pipelined, with one-cycle latency for each operation.

\noindent\textbf{Local Cache:} The Local Cache stores the polynomial coefficients of recently accessed segments using a Least Recently Used (LRU) replacement policy. \change{Each cache row is configured to hold four 32-bit floating-point coefficients corresponding to the polynomial degrees $(Z, L, Q, C)$. 
It takes one cycle to match the segment index generated by the Breakpoint Comparator with the tag stored in the Local Cache. 
If a hit occurs in the subsequent cycle,} the matched entry directly supplies the stored coefficients to the Execution Unit, bypassing the Global Multi-Level LUT, as shown in the pipeline cycle in Figure~\ref{fig:hw_arch}.

On a cache miss, \change{the lane arbiter checks the current status of the Global Multi-Level LUT. If the current lane is permitted to request data, the segment index is forwarded to the Global Multi-Level LUT, and the data is output in the next cycle. If the Global Multi-Level LUT is occupied by other lanes, the lane processing stalls up to three cycles (number of lanes minus one due to round-robin arbitration) to gain access. 
Once access to the Global Multi-Level LUT is granted, the coefficients are forwarded and the Local Cache is updated concurrently to hide the update latency from the critical path.}

\noindent\textbf{Global Multi-Level LUT:}
The Global Multi-Level LUT stores all coefficients for supported functions across four independent memory banks, one per degree $d \in \{0, 1, 2, 3\}$. 
Each bank has a fixed 32-bit width to store the coefficients. 
Our parameterized design allows instantiating this unit with different storage capacities.
\change{The implementation used in the experimental evaluations can store 20 constant coefficients 
since all segments have a constant term regardless of their degrees. 
The banks for the linear, quadratic, and cubic segments are designed to store 18, 12, and 3 coefficients, respectively.
We note that the Uni-SFU framework can support implementations with arbitrary sizes.}
Once a lane gains access to the global LUT, the full 3-cycle access path begins with index decoding and finding the address of the required coefficient in each memory bank.
In the subsequent cycle, the coefficients are fetched from the selected LUTs. 
Finally, they are propagated to the corresponding Execution Unit through the lane demultiplexer. 
The Local Cache is updated simultaneously during propagation to serve subsequent accesses. 
This mechanism ensures that multi-stage polynomial evaluations can be performed without stalls.


\noindent\textbf{Execution Unit:}
The Execution Unit comprises three parallel 32-bit floating-point MAC–ADD units, one for the cubic, one for the quadratic, and another for the linear term.
The input and coefficients are delivered to all active MAC--ADD units simultaneously.
Each MAC--ADD unit takes two cycles to compute results and add the constant term. That is, a first degree (linear) segment's output is ready in two cycles, a quadratic segment output is ready in four cycles, and a cubic segment output is ready in six cycles.

\change{The proposed Uni-SFU framework is scalable and precision agnostic, allowing each functional unit to be configured for low-bit precisions. For instance, the Breakpoint Comparator can be scaled down to 16-bit or 8-bit configurations to adapt to low precision SFUs. Similarly, the Local Cache, Global Multi Level LUT, and Execution Unit can be customized for low-bit implementations. The accuracy under such low bit configurations is validated in Section~\ref{subsec:VI-E} by casting the 32 bit breakpoints and coefficients, an approach that is demonstrated to be highly robust across both numerical accuracy and hardware implementations.}

\subsection{Uni-SFU Latency Analysis} \label{subsec:latency}

The end-to-end latency of Uni-SFU consists of 
(1) Breakpoint Comparator ($L_{\text{cmp}}$), (2) Local Cache ($L_{\text{cache}}$),
(3) Global Multi-Level LUT ($L_{\text{lut}}$), and (4) Execution Unit ($L_{\text{exec}}$).

\noindent\textbf{Best-Case Latency:} The Uni-SFU achieves the shortest latency $L_{\text{best}}$ when a \textit{hit occurs in the Local Cache}. 
In this case, the latency breakdown is as follows:
\begin{itemize}
    \item $L_{\text{cmp}} = 2$ cycles to determine the segment index by using a priority encoding scheme.
    \item $L_{\text{cache}} = 2$ cycles to retrieve the coefficients on a hit.
    \item $L_{\text{lut}} = 0$ cycles since it is bypassed.
    \item $L_{\text{exec}}$ is the time (2--6 cycles) taken by the MAC-ADD units in the Execution Unit.
\end{itemize}
%
Therefore, $L_{\text{best-case}} = 4 + L_{\text{exec}}$ cycles, where $L_{\text{exec}}$ is 2 cycles for linear, 4 cycles for quadratic, and 6 cycles for cubic segments.
Once the pipeline is filled, the Uni-SFU can reach one output/lane/cycle if the inputs result in a hit in the Local Cache.

\noindent\textbf{Worst-Case Latency:} The worst-case occurs when \textit{the segment misses in the Local Cache}, necessitating loading the coefficients from the global LUT with the following latency breakdown:
\begin{itemize}
    \item $L_{\text{cmp}} = 2$ cycles (identical to the best case).
    \item $L_{\text{cache}} = 1$ cycle for the parallel tag comparison to detect a cache miss.
    \item $L_{\text{lut}} = 6$ cycles: 3 cycles worst-case latency for winning the round-robin lane arbitration (this latency will be lower if the other lanes do not experience a miss in the same cycle). Then, 3 additional cycles within the LUT pipeline (address decoding, data fetching, and propagating to the Execution Unit). 
    \item $L_{\text{exec}}$ is the time (2--6 cycles) taken by the MAC-ADD units in the Execution Unit (same as the best-case).
\end{itemize}

Therefore, the total worst-case latency is $L_{\text{worst-case}} = 10 + L_{\text{exec}}$ cycles, incurring a conditional penalty of up to 5 additional cycles compared to the best case. 

\subsection{Locality Profiling and Average Latency}\label{subsec:V-D}



Figure~\ref{fig:locality_rate} shows the hit rate for various commonly used CNN, Transformer and NLP models when the proposed Uni-SFU architecture with a Local Cache of 4 entries evaluates their activation functions. 
Various representative layers of the CNN architectures exhibit higher locality up to 90\%, with an average hit rate of 83\%.
The Transformer architectures experience relatively lower (72\%--82\%) hit rates due to the homogeneous input distributions in convolutional layers. \change{\change{For NLP-specific workloads, the local-cache hit rate varies significantly depending on the model architecture. While LLaMA-2 7B\cite{touvron2023llama2} achieves an average hit rate of 97.2\%, GPT-Neo 1.3B\cite{gpt-neo} maintains a functional average hit rate of 72.7\% to sustain pipeline throughput.
}} 
In both vision and NLP families, locality decreases modestly from earlier to later layers as activation distributions broaden with network depth, yet remains above 72\% across all profiled cases.
These profiling results suggest that a small number of cache entries is sufficient to capture the dominant working set of segments across diverse inference workloads.

\noindent\textbf{Average Latency:} The average latency depends on both the cache hit rate and the degree distribution of the active segments. Given a hit rate $h$, it can be approximated as:
\begin{equation}
L_{\text{avg}} = h \cdot L_{\text{best}} + (1-h) \cdot L_{\text{worst}}.
\end{equation}
In our studies, we consistently observe 72\%--90\% hit rates, indicating that average latency approaches the best-case value in practice. A comprehensive analysis of hit rates across diverse workloads is shown in Section~\ref{subsec:V-D} and cache-entry configurations is presented in Section~\ref{subsec:VI-C}.

\begin{figure}[t]
\centering
\change{
\vspace{-2mm}
\includegraphics[width=0.9\linewidth]{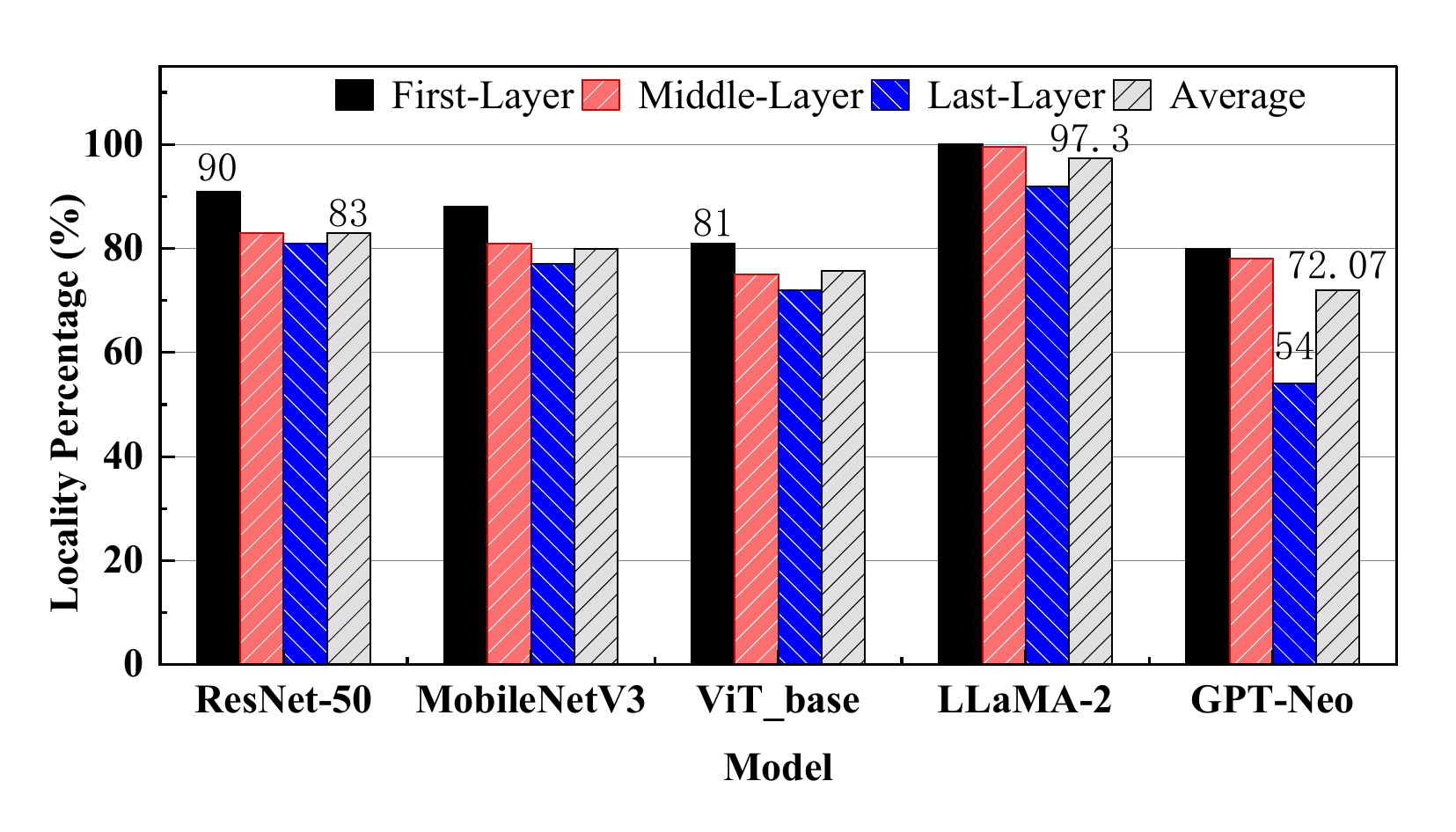}
\vspace{-3mm}
\caption{Locality rate insight for DNN, transformer neural network models and NLP workloads.}
\vspace{-2mm}
\label{fig:locality_rate}
}
\end{figure}

%% file: text/V-Implementation_and_Evaluation-sfu.tex
\begin{figure*}[t]
\centering
\includegraphics[width=1\linewidth]{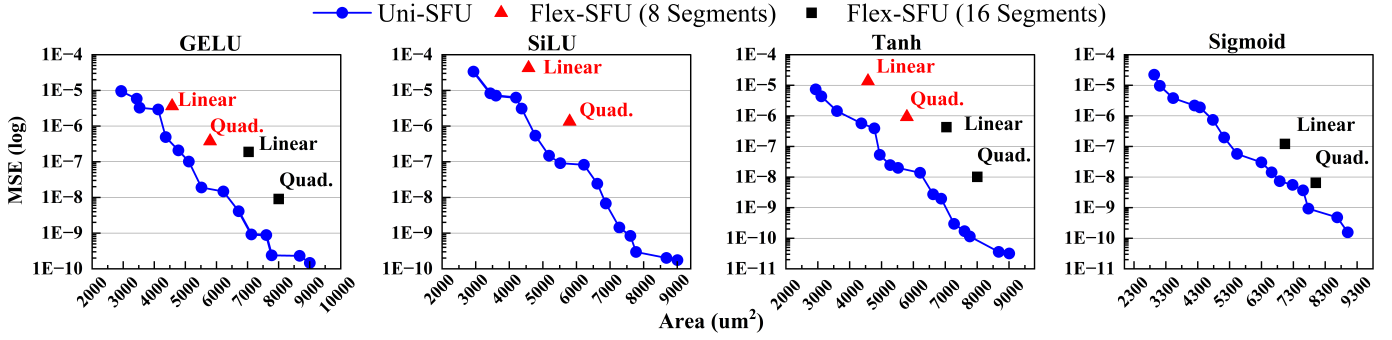}
\vspace{-6.5mm}
\caption{Accuracy and Area Comparison between Flex-SFU and Uni-SFU.}
\label{fig:pareto_vs_flex}
\vspace{-5mm}
\end{figure*}

\section{Experimental Results} \label{sec:VI}
\subsection{Experimental Setup}\label{subsec:VI-A}

\noindent\textbf{Hardware Implementation:} We implement the proposed Uni-SFU in Verilog HDL and synthesize it using Synopsys Design Compiler (\texttt{compile\_ultra}) targeting GF~22nm CMOS technology at 500~MHz. All operations use 32-bit IEEE~754 FP32 arithmetic. 
Single-lane and multi-lane configurations (up to 16 lanes) are synthesized independently to characterize area scaling. 
Segment boundaries and polynomial coefficients are preloaded into the Global Multi-Level LUT, implemented as a register-based structure in the GF~22nm standard cell library.
Area results are collected from post-synthesis netlist reports with no timing violations. 

\noindent\textbf{Workloads and Simulation Platform:}
Uni-SFU is evaluated using the TIMM package~\cite{rw2019timm}. From over 1,200 models, we select 700 representative benchmarks by removing architectural redundancies while preserving model diversity. For accuracy evaluation, original activation functions are replaced with the proposed piecewise-polynomial approximations in PyTorch without modifying model architectures. Experiments are conducted on an HPC server with AMD EPYC 7343 processors and an NVIDIA RTX A4000 GPU. Evaluating the full model suite takes approximately 84 hours for inference, while training each model requires roughly 6 hours.


\subsection{Area-Throughput Trade-Off Optimization}\label{subsec:VI-B}
\change{Before detailed SOTA comparisons, this section evaluates the area-throughput trade-offs using GELU, SiLU, Tanh, Sigmoid, Softplus, and ELU activation functions used by the most relevant prior work~\cite{reggiani2023flex}. Figure~\ref{fig:pareto_vs_flex} shows our Pareto-optimal results from the per-function optimization described in Section~\ref{subsec:IV-C}.} 
The blue curves with circle markers show optimal designs found by Uni-SFU by sweeping the segment count from 1 to 16. The red and black markers represent 8- and 16-segment points from Flex-SFU~\cite{reggiani2023flex}. To ensure a fair comparison, Flex-SFU is evaluated by mapping its configuration onto our hardware model. 
Specifically, uniform LUTs and corresponding Execution Units are integrated based on its original design specifications and maximum polynomial degrees. 
These results indicate that Uni-SFU achieves both lower approximation error and area at the same time. 
For GELU, Uni-SFU finds design points that have 9.95$\times$ lower MSE than the Flex-SFU design with quadratic approximation of 16 segments within the same area footprint. Similarly, our SiLU implementation achieves 10.7$\times$ lower MSE than the Flex-SFU's quadratic approximation of 8 segments when using the same area. Conversely, Uni-SFU achieves the same MSE with a 23\% smaller area compared to the 8-segment quadratic approximation of Flex-SFU for SiLU. We also observe similar substantial benefits for Tanh and Sigmoid, as shown in Figure~\ref{fig:pareto_vs_flex}. 
These results show that mixed-degree assignment and nonuniform segmentation co-optimize the area footprint and approximation accuracy effectively.

\subsection{Hardware Parameter Optimization}\label{subsec:VI-C}
The proposed Uni-SFU architecture can be implemented using different Local Cache sizes and lane counts. 
This section studies their area-performance trade-off.

\noindent\textbf{Cache Entry Analysis:} Figure~\ref{fig:cache_hit} illustrates the cache hit rate for a 4-lane configuration obtained for the DNNs in the TIMM library (covering 700+ variants, excluding intra-family scaling). The hit rate improves sharply from 35\% at a single entry to 82\% at 4 entries (marked in red), after which the curve exhibits diminishing returns and saturates toward 100\% at 16 entries. This suggests that 4 entries effectively capture the active working set for most mainstream inference workloads.

\begin{figure}[t]
\centering
\vspace{-1mm}
\includegraphics[width=1\linewidth]{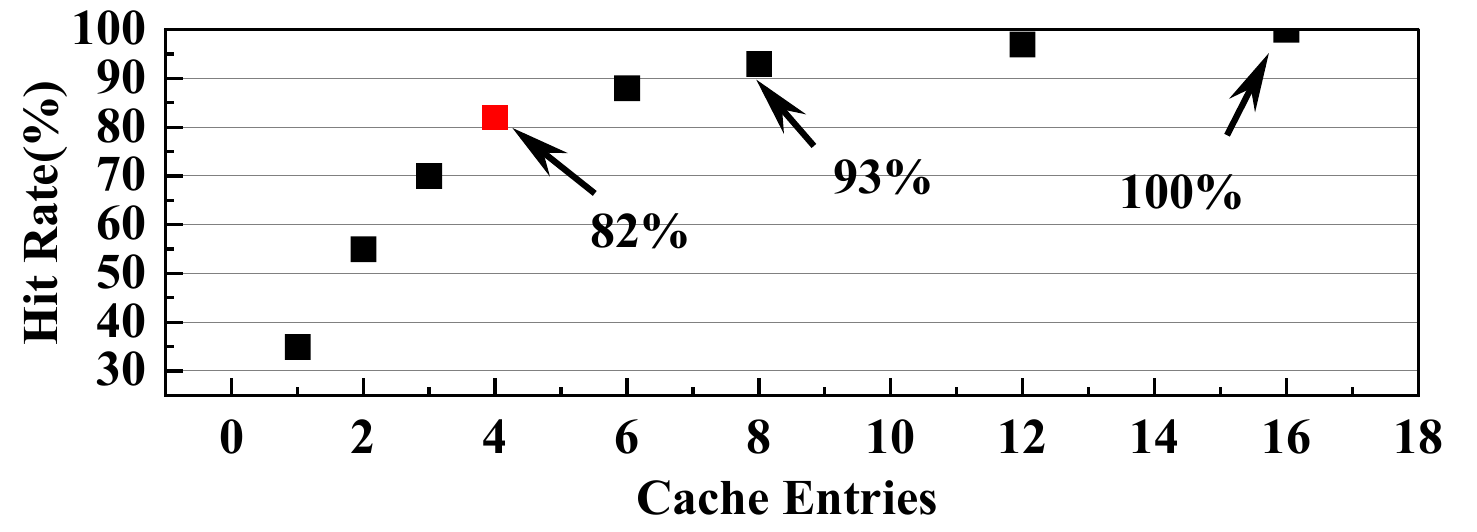}
\vspace{-7mm}
\caption{Cache Hit Rate vs. Cache Entries.}
\label{fig:cache_hit}
\vspace{-3mm}
\end{figure}

\begin{table*}[t]
\centering
\change{
\caption{A Comprehensive Hardware and Performance Comparison with SOTA SFU Designs.}
\label{tab:sota_comparison}
\resizebox{\textwidth}{!}{%
\begin{tabular}{lcccccccc}
\toprule
\textbf{Work} & \textbf{Supported Functions} & \textbf{Degree} & \textbf{Prec.} & \textbf{Area ($\mu\text{m}^2$)} & \textbf{Acc. (MSE/rMAE/PPL)$^\ddagger$} & \textbf{Share} & \textbf{Power (mW)} & \textbf{Platform} \\
\midrule
PACE~\cite{prasad2025pace} & Hardswish, SiLU, GELU, Sigm, Exp, Sin, Cos & Fixed (1-3) & FP32/16, BF16, FP8 & 10538 / 210760 & $1.2\times 10^{-6}$ (MSE) & Yes & N.F. & VPU \\ \addlinespace
QPA~\cite{geng2023qpa} & Sin, Exp, Sinc, Sigm, Tanh, Log2, Softsign & Fixed (1-2) & 8/12/16-bit & 4396.4 / 2763.3 & $2.51\times 10^{-3}$ (rMAE) & No & 1.94 / 1.66 & VPU \\ \addlinespace
PACE-Lite~\cite{prasad2025lite} & GELU, SiLU & Fixed (1-3) & FP32/16, BFP16 & 28178 / 563560 & Acc. Drop $<$ 1\% & Yes & N.F. & VPU \\ \addlinespace
Lu et al.~\cite{lu2023efficient} & GELU, Softmax, LayerNorm & Fixed (1) & INT8 & N.F. & N.F. & No & N.F. & FPGA \\ \addlinespace
Marca~\cite{li2024marca} & Exp, Log & Var (1-2) & FP32 & 28.20 mm$^2$ / N.F. & PPL incr. $<$ 0.3 & Yes & 3.92 / 3.36 & GPU \\ \addlinespace
GELU-MSDF$^\dagger$~\cite{taghavizade2024gelu} & GELU Only & Fixed (2) & INT32 & 8775 / 1930.5 & N.F. & No & 10.71 / 6.59 & AI Accelerator \\ \addlinespace
PEANO-ViT~\cite{sadeghi2024peano} & LN, Softmax, GELU, Rsqrt, Reciprocal & Var (1-2) & FP32-16 & N.F. & $2.56\times 10^{-4}$  & Yes & N.F. & FPGA \\ \addlinespace
FlexSFU (8 seg)~\cite{reggiani2023flex} & Tanh, Sigm, SiLU, GELU & Var (1-2) & FP32/16/INT8 & 7983.0 / 8001.2$^\mathsection$ & $1.35\times 10^{-6}$ (sq-AAE$^\star$) & No & 0.7 / 0.5998 & VPU \\ \addlinespace
FlexSFU (16 seg)~\cite{reggiani2023flex} & Tanh, Sigm, SiLU, GELU & Var (1-2) & FP32/16/INT8 & 9490.6 / 8967.6$^\mathsection$ & $3.35\times 10^{-8}$ (sq-AAE$^\star$) & No & 0.9 / 0.7712 & VPU \\
\midrule
\textbf{Uni-SFU (8 seg)} & \multirow{2}{*}{\begin{tabular}[c]{@{}l@{}}GELU, SiLU, Sigmoid,\\ Tanh, Softplus, ELU\end{tabular}} & \multirow{2}{*}{\textbf{Hybrid}} & \multirow{2}{*}{\textbf{FP32}} & \textbf{6204} & $\mathbf{1.56\times 10^{-8}(MSE)}$ & \multirow{2}{*}{\textbf{Yes}} & \textbf{0.56} & \multirow{2}{*}{\textbf{Universal}} \\ 
\textbf{Uni-SFU (16 seg)} & & & & \textbf{6800} & $\mathbf{1.85\times 10^{-9}(MSE)}$ & & \textbf{0.63} & \\
\bottomrule
\end{tabular}%
}
}
\begin{minipage}{\textwidth}
\scriptsize
\change{
\setlength{\parskip}{0.5pt}
$^\dagger$ Metrics for \textbf{GELU-MSDF} (45nm), \textbf{QPA} (28nm) and \textbf{PACE/PACE-Lite} (12nm) are normalized to 22nm technology node. \\
$^\ddagger$ Accuracy metrics are reported based on the \textit{worst-case error} observed across all supported functions for each respective work. \\
$^\mathsection$ For \textbf{FlexSFU}, the normalized area represents the hybrid estimated total by scaling storage area (factor 1.56~\cite{sarangi2021deepscaletool}) and synthesized Execution Units. \\
$^\star$ \textbf{sq-AAE}: Quadratic approx. via AAE; \textbf{MSE}: Mean Squared Error; \textbf{rMAE}: Real Max Absolute Error; \textbf{PPL}: Perplexity. \\
\text{N.F.}: Not Found.
}
\vspace{-3mm}
\end{minipage}

\end{table*}

\begin{figure}[b!]
\centering
\vspace{-3mm}
\includegraphics[width=1\linewidth]{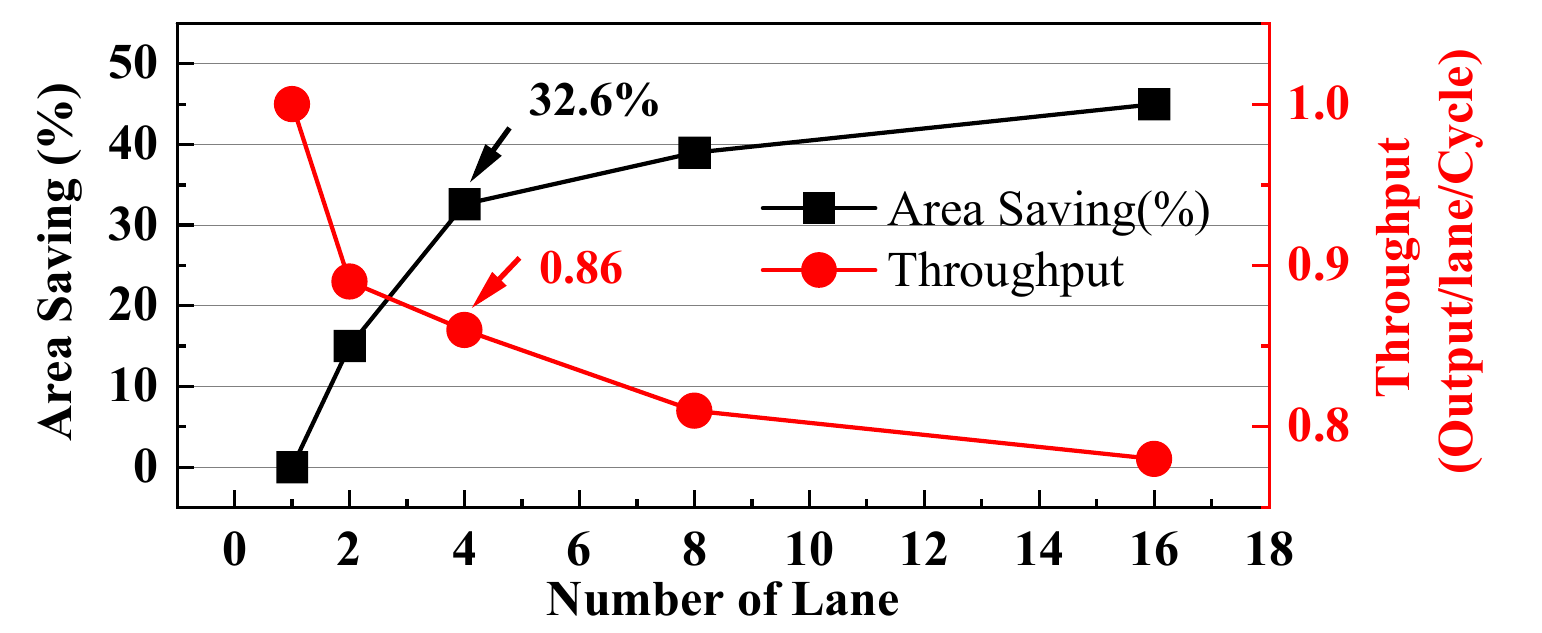}
\vspace{-7mm}
\caption{Shared Lane Number vs. Area Saving Percentage.}
\label{fig:throught_put}
\end{figure}

\noindent\textbf{The Number of Lanes vs. Area Saving and Throughput:} 
Uni-SFU's multi-lane architecture enables a cost-efficient mechanism to increase the parallelism (i.e., throughput) 
by sharing the resource-hungry global LUT across multiple lanes. 
\change{A \textit{baseline design} would replicate all blocks, including the global LUT to make the throughput proportional to the number of lanes at the cost of design area.}
Figure~\ref{fig:throught_put} compares the proposed Uni-SFU design against this baseline as a function of the number of lanes. 
The designs are identical when there is only one lane: no area savings (left axis) and the same throughput (right axis). 
When the lane count increases to two, Uni-SFU provides around 15\% lower area compared to full replication, while still achieving 0.9$\times$ of the peak throughput. 
As the number of lanes increases, we observe a diminishing rate of return in area saving (left axis, $\blacksquare$ markers) and throughput reduction (right axis, ${\color{red}\bullet}$ markers). The throughput degrades due to contention for the shared global LUT when cache misses occur.
This trend shows that our multi-lane design has increasing efficiency as the global LUT is amortized over an increasing number of lanes. 
To maintain a high throughput, our final design uses four lanes, which results in 0.86$\times$ peak throughput with 32.6\% area savings.

\subsection{\change{Comparisons Against Prior Approaches}}\label{subsec:VI-D}
\change{The proposed Uni-SFU achieves superior accuracy across all supported activation functions compared to SOTA SFU designs, with a maximum MSE of $1.85 \times 10^{-9}$ (16-segment) and $1.56 \times 10^{-8}$ (8-segment), outperforming all compared works by at least two orders of magnitude, as summarized in Table~\ref{tab:sota_comparison}. 
In terms of power consumption, Uni-SFU consumes 0.5607\,mW and 0.6317\,mW for the 8- and 16-segment configurations respectively, which represents a massive power reduction compared to alternative SOTA FP32-precision designs under identical normalization (e.g., Marca: 3.36\,mW).} 

\change{We note that our particular Uni-SFU implementation supports full FP32 precision and a much broader function set, while the general Uni-SFU framework can be applied to lower-precision computation (e.g., INT8 or FP8 quantization). Lower-precision computation is widely adopted to accelerate DNN models, but high-precision SFU support remains essential. In modern mixed-precision inference pipelines, complex non-linear activation functions (such as GELU and SiLU) are typically evaluated in FP32 precision via local dequantization to preserve model robustness and prevent severe accuracy degradation \cite{yuan2025give, tensorrt}. Advanced quantization frameworks directly demonstrate that special functional layers can implement a mixed data flow, where quantized tensors are temporarily cast back to higher precision for precise execution \cite{xi2024jetfire}.
Finally, Uni-SFU is the only universal design supporting multi-function hardware sharing on a single unified unit for six distinct activation functions: GELU, SiLU, Sigmoid, Tanh, Softplus, and ELU.}

\begin{table}[b!]
\centering
\caption{Area Comparison between Flex-SFU~\cite{reggiani2023flex} (Quadratic Degree) and Uni-SFU (Mixed Degree). $\%\downarrow$ denotes the area reduction of Uni-SFU relative to Flex-SFU.}
\label{tab:area_comparison}
\small
\setlength{\tabcolsep}{3.6pt}
\renewcommand{\arraystretch}{0.9}
\begin{tabular}{lccc}
\toprule
\textbf{Configuration} & \textbf{Flex-SFU ($\mu m^2$)} & \textbf{Uni-SFU ($\mu m^2$)} & \textbf{vs. Flex-SFU} \\
\midrule
\textit{Single-lane} & & & \\
~~8-segment  & 8,001.2 & 6,204.0 & 22.5\%$\downarrow$ \\
~~16-segment & 8,967.6 & 6,800.0 & 24.2\%$\downarrow$ \\
\midrule
\textit{4-lane} & & & \\
~~8-segment  & 20,410.6 & 17,251.1 & 15.5\%$\downarrow$ \\
~~16-segment & 22,238.1 & 18,347.1 & 17.5\%$\downarrow$ \\
\bottomrule
\end{tabular}
\end{table}

\subsection{Comparison with Flex-SFU}\label{subsec:VI-E}
\change{Among all the SOTA designs, Flex-SFU has the highest performance and the lowest power consumption. Hence, we further benchmark Uni-SFU with respect to Flex-SFU in terms of area, accuracy, and power consumption.}

\noindent \textbf{Area Comparison:} Flex-SFU~\cite{reggiani2023flex} reports storage areas of 7,983~$\mu\text{m}^2$ and 9,491~$\mu\text{m}^2$ for 8- and 16-segment configurations in 28\,nm CMOS, respectively. 
For a fair comparison, we scale its storage area to 22\,nm using a factor of 1.56~\cite{sarangi2021deepscaletool}, 
and include our synthesized Execution Units, yielding estimated total areas of 8,001.2~$\mu\text{m}^2$ and 8,967.6~$\mu\text{m}^2$ (Table~\ref{tab:area_comparison}). In contrast, Uni-SFU occupies only 6,204.0~$\mu\text{m}^2$ and 6,800.0~$\mu\text{m}^2$ for the respective cases. This architectural optimization achieves a consistent area reduction of 22.5\%--24.2\% for single-lane and up to 17.5\% for 4-lane shared configurations. 

The 4-lane baseline for Flex-SFU is estimated by accounting for its reported area scaling factors associated with multi-ported memory structures.
Table~\ref{tab:area_breakdown} summarizes the area breakdown of our single- and 4-lane designs. 
The Breakpoint Comparator and Local Cache consume moderate (6.9\% and 5.7\%) area of the single-lane design. 
Hence, replicating them in the multi-lane design offers a good trade-off, increasing their contribution to 10.3\% and 8.4\%, respectively. 
In contrast, the Global Multi-level LUT takes up 45\% of the single-lane design, justifying our motivation behind the multi-lane architecture. 
Indeed, its area contribution drops to 16.7\% in the 4-lane design. 
The Execution Unit takes up 42.4\% (the second largest) of the single lane, and 62.9\% (the largest) of the 4-lane design since it is replicated. 
Finally, the glue logic (e.g., multi-lane arbitration and demux to the Execution Unit) is not needed for the single-lane design, while its overhead in our 4-lane design is only 1.7\%, showing the efficiency of the proposed multi-lane architecture. 

\begin{table}[t]
\centering
\caption{Area breakdown percentage of Uni-SFU hardware components.}
\label{tab:area_breakdown}
\setlength{\tabcolsep}{6pt}
\renewcommand{\arraystretch}{0.9}
\small
\begin{tabular}{lcc}
\toprule
\textbf{Hardware Component} & \textbf{Single-lane (\%)} & \textbf{4-lane Shared (\%)} \\
\midrule
Breakpoint Comparator & 6.9\%  & 10.3\% \\
Local Cache           & 5.7\%  & 8.4\%  \\
Global ML-LUT         & 45.0\% & 16.7\% \\
Execution Units       & 42.4\% & 62.9\% \\
Glue Logic            & --     & 1.7\%  \\
\bottomrule
\end{tabular}
\end{table}

\begin{table}[t]
\centering
\vspace{-2mm}
\caption{MSE comparison with Flex-SFU~\cite{reggiani2023flex}. Here, $\mathrm{x}\downarrow$ and $\mathrm{x}\uparrow$ denote lower and higher than the reference by that factor.}
\label{tab:mse_comparison}
\small
\setlength{\tabcolsep}{3pt}
\renewcommand{\arraystretch}{1}
\begin{tabular*}{\columnwidth}{@{\extracolsep{\fill}}lccccc@{}}
\toprule
\textbf{Func.} &
\textbf{Lin. Ref.} &
\textbf{Quad. Ref.} &
\textbf{Uni-SFU} &
\textbf{vs Lin.} &
\textbf{vs Quad.} \\
\midrule
\multicolumn{6}{c}{\textbf{8-segment case} ($u=[\,1,\,2,\,2,\,3\,]$)} \\
\midrule
GELU     & $3.65\mathrm{e}{-6}$ & $3.78\mathrm{e}{-7}$ & $\mathbf{1.56e{-8}}$ & $234.3\mathrm{x}\downarrow$  & $24.3\mathrm{x}\downarrow$ \\
Sigmoid  & --                    & --                    & $\mathbf{3.23e{-9}}$ & --                            & --                         \\
SiLU     & $4.27\mathrm{e}{-5}$ & $1.35\mathrm{e}{-6}$ & $\mathbf{8.22e{-8}}$ & $519.5\mathrm{x}\downarrow$ & $41.0\mathrm{x}\downarrow$ \\
Tanh     & $1.37\mathrm{e}{-5}$ & $9.28\mathrm{e}{-7}$ & $\mathbf{2.15e{-8}}$ & $637.8\mathrm{x}\downarrow$  & $16.4\mathrm{x}\downarrow$ \\
Softplus & --                    & --                    & $\mathbf{7.09e{-9}}$ & --                            & --                         \\
ELU      & --                    & --                    & $\mathbf{3.63e{-10}}$ & --                           & --                         \\
\midrule
\multicolumn{6}{c}{\textbf{16-segment case} ($u=[\,2,\,6,\,9,\,3\,]$)} \\
\midrule
GELU     & $1.89\mathrm{e}{-7}$ & $9.07\mathrm{e}{-9}$ & $\mathbf{5.62e{-10}}$ & $336.1\mathrm{x}\downarrow$ & $16.1\mathrm{x}\downarrow$ \\
Sigmoid  & $2.88\mathrm{e}{-7}$ & $6.50\mathrm{e}{-9}$ & $\mathbf{7.40e{-10}}$ & $389.3\mathrm{x}\downarrow$ & $8.8\mathrm{x}\downarrow$  \\
SiLU     & --                    & --                    & $\mathbf{3.21e{-9}}$  & --                            & --                         \\
Tanh     & $4.26\mathrm{e}{-7}$ & $1.02\mathrm{e}{-8}$ & $\mathbf{1.85e{-9}}$  & $230.9\mathrm{x}\downarrow$ & $5.5\mathrm{x}\downarrow$  \\
Softplus & --                    & --                    & $\mathbf{4.53e{-9}}$  & --                            & --                         \\
ELU      & --                    & --                    & $\mathbf{4.10e{-10}}$ & --                           & --                         \\
\bottomrule
\multicolumn{6}{l}{\footnotesize ``--'' indicates data not reported in the related work.}
\end{tabular*}
\vspace{-2mm}
\end{table}

\noindent \textbf{Accuracy Comparison:} Uni-SFU achieves a smaller area than Flex-SFU while consistently providing higher accuracy, as summarized in Table~\ref{tab:mse_comparison}, with breakpoints re-optimized following Section~\ref{subsec:IV-E}. While Flex-SFU uses the Square of Average Absolute Error (sq-AAE), we report MSE for consistency with our optimization objective. Uni-SFU also achieves $1.1\times$--$3\times$ lower error under sq-AAE. At 8 segments, Uni-SFU achieves up to $519.5\times$ and $41.0\times$ lower error than the linear and quadratic references, respectively; at 16 segments, the reductions reach $389.3\times$ and $16.1\times$. These results demonstrate a better accuracy-area trade-off from mixed-degree polynomial assignment than the uniform-degree design in Flex-SFU.

To analyze scalability, we evaluate function approximation accuracy by casting fitted segment breakpoints and coefficients directly into the FP16 domain. 
As summarized in Table~\ref{tab:precision_mse}, Uni-SFU maintains high functional fidelity even under reduced bit-width configuration. Notably, while the baseline PEANO-ViT~\cite{sadeghi2024peano} is restricted to a narrow domain range of $[-4, 4]$, our Uni-SFU provides broader coverage of $[-8, 8]$. Despite this significantly wider approximation range, Uni-SFU achieves order-of-magnitude lower MSE compared to PEANO-ViT. This result demonstrates that our underlying algorithmic framework achieves superior approximation fidelity over a larger input space while maintaining robust precision, confirming its scalability and efficiency for diverse deployment scenarios.

\noindent \change{\textbf{Power Comparison:} Uni-SFU consumes 0.5607 mW and 0.6317 mW under the 8-segment and 16-segment configurations, respectively. Compared to the corresponding second-order Flex-SFU designs (0.7 mW and 0.9 mW), Uni-SFU achieves 19.9\% and 29.8\% power savings, respectively, while maintaining superior approximation accuracy. For example, Uni-SFU approximates the SiLU function with an MSE of $8.22\times10^{-8}$, which is significantly lower than the MSE of $1.35\times10^{-6}$ achieved by Flex-SFU.}

\begin{table}[t]
\change{
\vspace{-2.5mm}
\caption{MSE Comparison in FP16 Domain.}
\label{tab:precision_mse}
\centering
\small
\resizebox{\linewidth}{!}{
\begin{tabular}{lcccc}
\toprule
\textbf{Methodology / Work} & \textbf{Segments} & \textbf{Range} & \textbf{Precision} & \textbf{Max. MSE} \\
\midrule
PEANO-ViT~\cite{sadeghi2024peano} & 7 & $[-4, 4]$ & FP16 & $2.65 \times 10^{-4}$ \\
PEANO-ViT~\cite{sadeghi2024peano} & 10 & $[-4, 4]$ & FP16 & $8.31 \times 10^{-5}$ \\
\textbf{Uni-SFU} & 8 & $[-8, 8]$ & FP16 & $\mathbf{2.5 \times 10^{-6}}$ \\
\textbf{Uni-SFU} & 16 & $[-8, 8]$ & FP16 & $\mathbf{4.1 \times 10^{-7}}$ \\
\bottomrule
\end{tabular}%
}
}
\end{table}

\begin{figure}[b]
\centering
\includegraphics[width=1\linewidth]{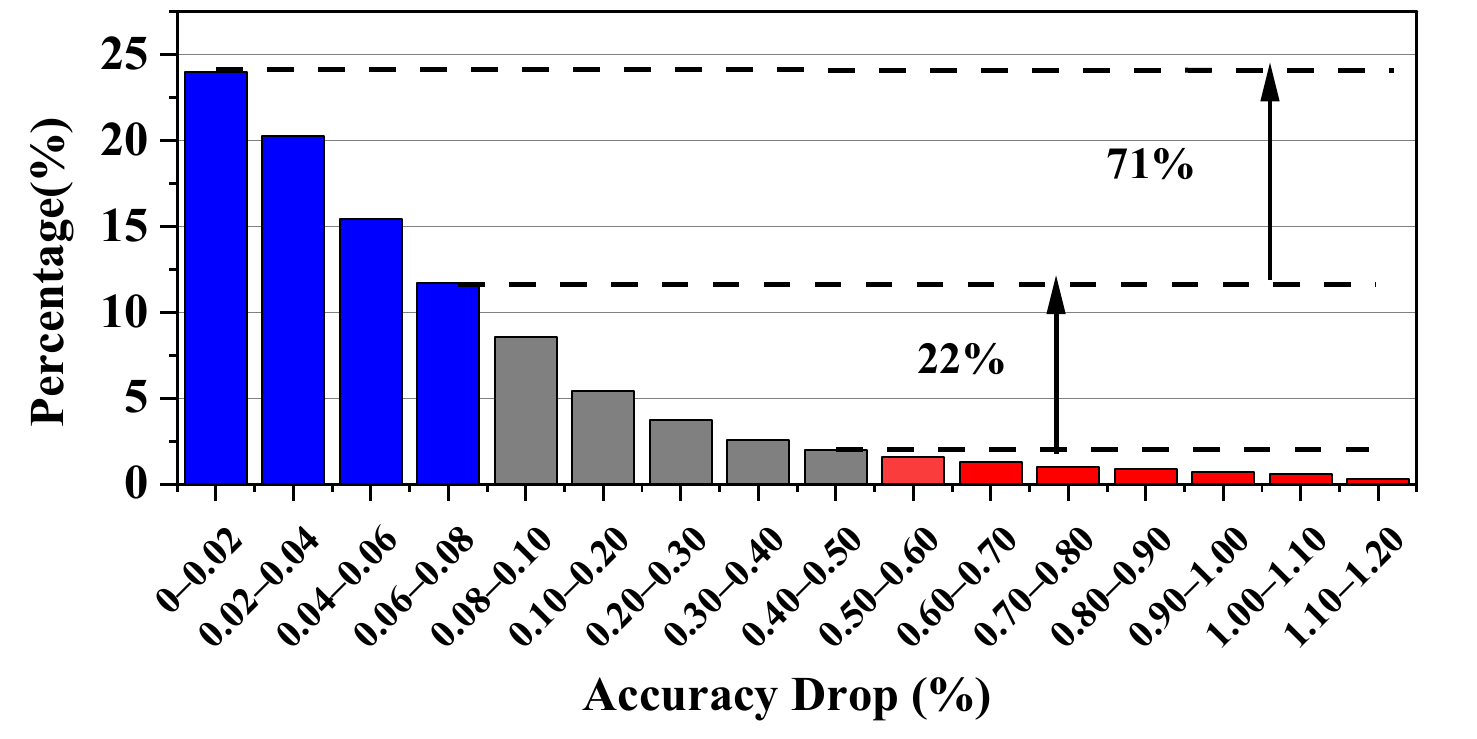}
\vspace{-7mm}
\caption{Distribution of dataset accuracy drop across 700 DNNs.}
\label{fig:cdf}
\vspace{-4mm}
\end{figure}

\begin{figure*}[t]
\change{
\centering
\vspace{-2mm}
\includegraphics[width=1.0\linewidth]{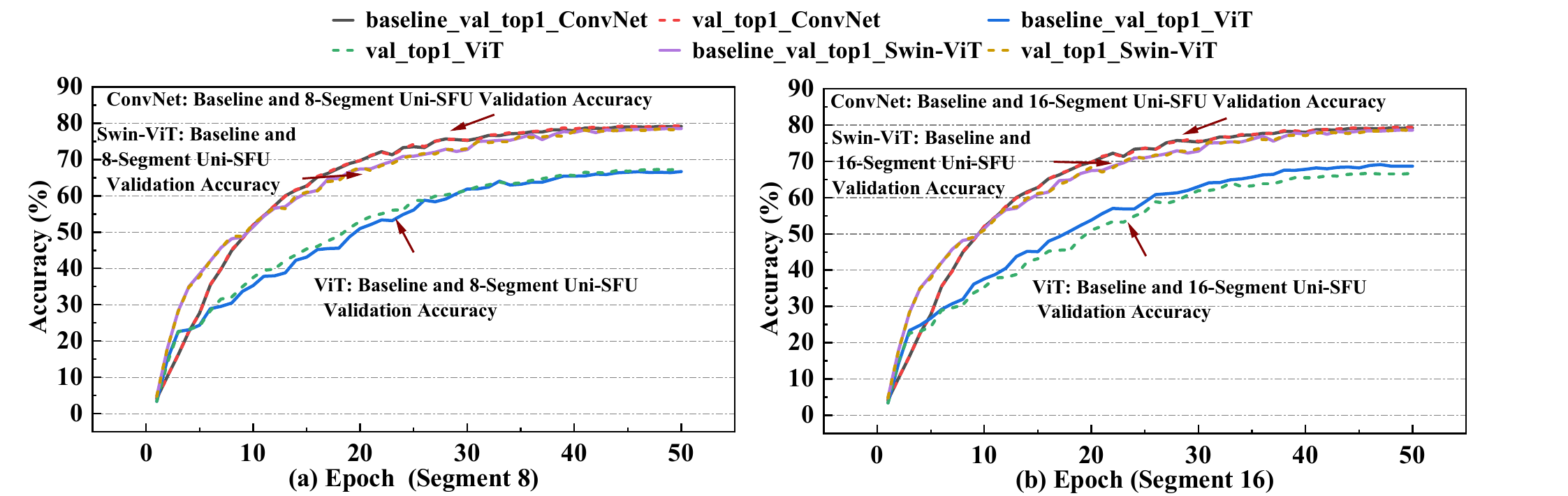}
\caption{Validation accuracy during training comparison with original non-linear activation Functions.}
\label{fig:training}
}
\vspace{-1mm}
\end{figure*}

\vspace{-1mm}
\begin{table}[t]
\centering
\renewcommand{\arraystretch}{1.1}
\caption{Mean accuracy drop on DNN workloads across diverse activation combinations.}
\label{tab:acc_drop}
\begin{tabular}{llrc}
\toprule
Activation Combination & Type & Count & Mean Drop (\%) \\
\midrule
GELU / SiLU / Sigmoid / Tanh & Single   & 490 & 0.012--0.078 \\
\midrule
Sigmoid+SiLU                 & 2-Comb   & 110 & 0.125 \\
GELU+Sigmoid                 & 2-Comb   & 55  & 0.198 \\
GELU+Sigmoid+SiLU            & 3-Comb   & 36  & 0.380 \\
GELU+Sigmoid+SiLU+Tanh       & 4-Comb   & 9   & 1.120 \\
\bottomrule
\end{tabular}
\end{table}

\subsection{End-to-End DNN Model Validation}\label{subsec:VI-F}
We finally evaluate the impact of Uni-SFU on the DNN model accuracy under both inference-only and full training scenarios to demonstrate its functional fidelity.

For inference evaluation, we replace the original activation functions with Uni-SFU approximations across both vision and NLP workloads. Figure~\ref{fig:cdf} shows the Top-1 accuracy degradation of Uni-SFU across 700 TIMM models. The vast majority of models experience a negligible accuracy drop below 0.02\% (the x-axis), with 71\% of the models having 0.08\% or lower degradation, which is negligible.
Only 7\% of the models have more than 0.50\% (highlighted as outliers in red) degradation, and
\textit{none of them experience more than 1.2\% accuracy degradation}.


Table~\ref{tab:acc_drop} summarizes results across 700 models. Single-function models maintain a mean accuracy drop below 0.078\%, while the most complex four-function combinations reach 1.12\%. Overall, the accuracy drop remains within 1.12\%, demonstrating the robustness of Uni-SFU across diverse DNN workloads.
We further evaluate three NLP models, GPT-Neo 1.3B~\cite{gpt-neo}, LLaMA-2 7B~\cite{touvron2023llama2}, and DistilBERT SST-2~\cite{sanh2019distilbert}, as shown in Table~\ref{tab:model_evaluation}. On WikiText-2~\cite{merity2016pointer}, Uni-SFU changes perplexity by only +0.0012 for GPT-Neo and -0.0005 for LLaMA-2. DistilBERT retains its original 91.06\% accuracy on GLUE SST-2~\cite{wang2019glue}. These results confirm that Uni-SFU preserves model accuracy across both vision and NLP workloads.

For training, we integrate Uni-SFU into both the forward and backward paths. Since the backward pass gradient calculation explicitly requires the activation function's derivative via the chain rule, maintaining high approximation fidelity is crucial to prevent gradient corruption. We validate Uni-SFU by training three vision models (ConvNext~\cite{todi2023convnext}, ViT~\cite{dosovitskiy2020image}, Swin-ViT~\cite{liu2021swin}) in TIMM, monitoring convergence over 50 epochs for 8-segment and 16-segment configurations. As shown in Figure ~\ref{fig:training}, validation accuracy tightly aligns with FP32 baselines; minor deviations are attributable to the learning rate scheduler rather than gradient errors. Final accuracy fluctuations remain strictly bounded: $-0.02\%$ to $+0.32\%$ for 8-segment, and $0.02\%$ to $0.18\%$ for 16-segment configurations. These results confirm that Uni-SFU ensures derivative fidelity and full network convergence.


\begin{table}[t]
\centering
\vspace{-2mm}
\change{
\caption{Model Evaluation on NLP workloads}
\label{tab:model_evaluation}
\resizebox{\columnwidth}{!}{%
\begin{tabular}{lcccccc}
\toprule
\textbf{Model} & \textbf{Activation} & \textbf{Test Set (Size)} & \textbf{Metric} & \textbf{Original} & \textbf{Approx} & $\boldsymbol{\Delta}$ \\ 
\midrule
\textbf{GPT-Neo 1.3B\cite{gpt-neo}} & \begin{tabular}[c]{@{}c@{}}GELU \end{tabular} & WikiText-2 & Perplexity & 12.3450 & 12.3462 & +0.0012 \\ 
\addlinespace 
\textbf{LLaMA-2 7B\cite{touvron2023llama2}} & \begin{tabular}[c]{@{}c@{}}SiLU \end{tabular} & WikiText-2 & Perplexity & 5.4323 & 5.4318 & -0.0005 \\ 
\addlinespace
\textbf{DistilBERT SST-2\cite{sanh2019distilbert}} & \begin{tabular}[c]{@{}c@{}}GELU\end{tabular} & \begin{tabular}[c]{@{}c@{}}GLUE\\ SST-2\end{tabular} & Accuracy (\%) & 91.06 & 91.06 & 0.00 \\ 
\bottomrule
\vspace{-5mm}
\end{tabular}%
}
}
\end{table}
\vspace{-5mm}

\change{\subsection{\change{Case Study: Platform Overhead and Accelerator Integration}}\label{subsec:VI-G}
We evaluate the impact of integrating Uni-SFU onto the open-source NVDLA accelerator platform \cite{nvdla}, configured with 16-bit precision and a 512-MAC PE array.}
\change{The legacy Single Data Processor (SDP) module accounts for 14\% of the NVDLA compute engine area when synthesized at the GlobalFoundries 22nm node. To support a 256-bit/cycle output, we instantiate four 4-lane shared Uni-SFU blocks, ensuring zero pipeline stalls. Our Uni-SFU cluster (0.0608 $mm^2$) replaces legacy LUTs within the SDP, achieving a 40\% area reduction in non-linear execution sub-modules. By retaining legacy logic for data alignment and handshaking, we ensure full functional compatibility. Normalized against the 22nm NVDLA (48.0 $mm^2$), the Uni-SFU cluster introduces a negligible 0.13\% global area overhead, keeping the combined SDP-related footprint below 10\% of the total engine. These results demonstrate that Uni-SFU provides efficient, high-throughput non-linear acceleration with minimal silicon area impact.
}

%% file: text/VI-Conclusion-sfu.tex
\section{Conclusion} \label{sec:VI}
This work presents Uni-SFU, a universal algorithm--hardware co-design framework supporting diverse activation functions on a single configuration. By employing non-uniform breakpoint search and mixed-degree polynomial assignment guided by an RTL-driven area model, Uni-SFU jointly optimizes accuracy and silicon area across all target functions. Evaluated across 700 CNNs and Transformers, as well as three NLP models using 22nm CMOS at 500 MHz, Uni-SFU limits the maximum top-1 accuracy degradation to 1.12\% for DNNs. \change{For NLP workloads, it demonstrates robust approximation fidelity, with a negligible perplexity variation of +0.0012 for GPT-Neo and -0.0005 for LLaMA-2.} In multi-lane deployment, a Local Cache exploits input locality, achieving a 0.85 hit rate and 32.6\% area reduction for 4-lane sharing (approaching 45\% at 16 lanes). Compared to prior work, Uni-SFU delivers $5.5\times$ to $637.8\times$ accuracy improvements at comparable area. This principled co-design generalizes resource distribution for multiple functions under a unified hardware envelope on resource-constrained platforms.
